\documentclass{article}

\usepackage[nonatbib, final]{neurips_2022}
\usepackage{tablefootnote}

\usepackage[utf8]{inputenc} 
\usepackage[T1]{fontenc}    
\usepackage{hyperref}       
\usepackage{url}            
\usepackage{booktabs}       
\usepackage{amsfonts}       
\usepackage{nicefrac}       
\usepackage{microtype}      
\usepackage{xcolor}         

\usepackage[utf8]{inputenc}
\usepackage{amsmath}
\usepackage{amssymb, color}
\usepackage{amsthm}
\usepackage{url}
\usepackage{mathtools}
\usepackage{algorithm}
\usepackage[noend]{algpseudocode}
\usepackage{caption}
\usepackage{subcaption}
\usepackage{xfrac}
\usepackage[acronym]{glossaries}
\usepackage{pgfplots}
\usepackage{adjustbox}

\newtheorem{theorem}{Theorem}[section]
\newtheorem{lemma}[theorem]{Lemma}

\newtheorem{corollary}[theorem]{Corollary}
\newtheorem{definition}[theorem]{Definition}

\newenvironment{myproof}[3]{
    \begin{proof}[\normalfont \textbf{Proof of #1 #2.}]
        #3
       \end{proof}
}

\algrenewcommand\algorithmicfunction{{}}

\title{Order-Invariant Cardinality Estimators Are Differentially Private}

\author{Charlie Dickens\\
  Yahoo\\
  \texttt{charlie.dickens@yahooinc.com}
  \And
  Justin Thaler\\
  Georgetown University\\
  \texttt{justin.thaler@georgetown.edu}
  \And 
  Daniel Ting \\ 
  Meta\footnote{This work was performed while the author was at Salesforce.} \\
  \texttt{dting@meta.com}
}

\newcommand{\ccal}{\ensuremath{\mathcal C}}
\newcommand{\dcal}{\ensuremath{\mathcal D}}

\newcommand{\kcal}{\ensuremath{\mathcal K}}

\newcommand{\scal}{\ensuremath{S}}

\newcommand{\Var}{\mathsf{Var}}

\newcommand{\defn}{:=}

\def\E{{\ensuremath{\mathbb E}}}

\newacronym{fm}{FM85}{Flajolet-Martin}
\newacronym{hll}{HLL}{HyperLogLog}

\begin{document}

\maketitle

\begin{abstract}
    We consider privacy in the context of streaming algorithms for cardinality estimation.
    We show that a large class of algorithms all satisfy $\epsilon$-differential privacy, 
    so long as (a) the algorithm is combined with a simple 
    down-sampling procedure, and (b) the input stream cardinality  
    is $\Omega(k/\epsilon)$. Here, $k$ is a certain parameter of the sketch
    that is always at most the sketch size in bits, but is typically much smaller.
    We also show that, even with no modification, algorithms in our
    class satisfy $(\epsilon, \delta)$-differential privacy,
    where $\delta$ falls exponentially with the stream cardinality. 
    Our analysis applies to essentially all popular cardinality estimation
    algorithms, and substantially generalizes and tightens privacy bounds from earlier works. 
    Our approach is faster and exhibits a better utility-space
    tradeoff than prior art.
\end{abstract}

\section{Introduction}

Cardinality estimation, or the distinct counting problem, is a fundamental data analysis task.
Typical applications are found in network traffic monitoring \cite{estan2003bitmap},
query optimization \cite{selinger1989access}, and counting unique search engine queries \cite{heule2013hyperloglog}.
A key challenge
is to perform this estimation in small space
while processing each data item quickly.
Typical approaches for solving this problem at scale involve data sketches such
as the \acrfull{fm} sketch \cite{flajolet1985probabilistic},
 \acrfull{hll} \cite{flajolet2007hyperloglog}, Bottom-$k$ 
 \cite{bar2002counting, cohen2007summarizing, beyer2009distinct}.
All these provide approximate cardinality estimates but use bounded space.

While research has historically focused on the accuracy, speed, and space usage of these sketches, recent work examines their privacy guarantees. These privacy-preserving properties have grown in importance as companies have built tools that can grant an appropriate level of privacy to different people and scenarios. The tools aid in satisfying users' demand for better data stewardship, 
while also ensuring compliance with regulatory requirements. 

We show that \emph{all} cardinality estimators in a class of 
hash-based, order-invariant sketches with bounded size are $\epsilon$-differentially private (DP) so long as the algorithm is combined with a simple 
down-sampling procedure and the true cardinality satisfies a mild lower bound. 
This lower bound requirement can be guaranteed 
to hold by inserting sufficiently many ``phantom elements'' into the stream when initializing the sketch. 
We also show that, even with no modification, algorithms in our
class satisfy $(\epsilon, \delta)$-differential privacy,
where $\delta$ falls exponentially with the stream cardinality.

Our novel analysis has significant benefits.
First, prior works on differentially private cardinality estimation have analyzed only specific sketches \cite{tschorsch2013algorithm, von2019rrtxfm, choi2020differentially, smith2020flajolet}. Moreover,
many of the sketches analyzed (e.g., \cite{tschorsch2013algorithm, smith2020flajolet}), while
reminiscent of sketches used in practice,
in fact
differ from practical sketches in important ways.
For example, 
Smith et al. 
\cite{smith2020flajolet} analyze
a \emph{variant} of \acrshort{hll} that Section \ref{s:hlls} shows has an update time that can be $k$ times slower than an \acrshort{hll} sketch with $k$ buckets.

While our analysis covers an entire class of sketches at once,
our error analysis improves upon prior work in many 
cases when specialized to specific sketches.
For example, our analysis yields tighter privacy bounds for HLL than the one given in \cite{choi2020differentially}, yielding both an $\epsilon$-DP guarantee, rather than
an $(\epsilon, \delta)$-DP guarantee, as well as tighter bounds on the failure probability $\delta$---see Section \ref{s:hlls} for details.
Crucially, the class of sketches
we analyze captures many (in fact, almost all to our knowledge)
of the sketches that are actually used in practice.
This means that
existing systems
can be used in contexts requiring 
privacy, either without modification
if streams are guaranteed to satisfy
the mild cardinality lower bound we require,
or with a simple pre-processing
step described if such cardinality lower bounds
may not hold. Thus, existing data infrastructure can be easily modified to provide DP guarantees, and in fact
existing sketches can be easily migrated to DP summaries.

\subsection{Related work}
\label{sec:related-work}
One perspective is that cardinality estimators cannot simultaneously preserve privacy and offer 
good utility \cite{desfontaines2019cardinality}.
However, this impossibility result applies only when an adversary 
However, this impossibility result applies only when an adversary can create and merge an arbitrary number of sketches, effectively observing an item's value many times. 
It does not address the privacy of one sketch itself.

Other works have studied more realistic models where either the hashes are public, but private noise is added to the sketch
\cite{tschorsch2013algorithm, mir2011pan, von2019rrtxfm},
or the hashes are secret \cite{choi2020differentially}
(i.e., not known to the adversary who is trying to ``break'' privacy). This latter setting turns out to permit less noisy cardinality estimates. 
Past works study specific sketches or a variant of a sketch.
For example, 
Smith et al. \cite{smith2020flajolet} 
show that an \acrshort{hll}-type sketch is
$\epsilon$-DP
while \cite{von2019rrtxfm} modifies the \acrshort{fm} sketch using coordinated sampling, which is also based on a private hash.
Variants of both models are analyzed by Choi et al. \cite{choi2020differentially},
and they show (amongst other contributions) a similar 
result to \cite{smith2020flajolet}, establishing that an \acrshort{fm}-type sketch is
differentially private.
Like these prior works, we focus on the setting when \emph{the hash functions are kept secret} from the adversary.
A related problem of differentially private estimation of cardinalities under set operations is studied by \cite{pagh2021efficient}, but they assume the inputs to each sketch are already de-duplicated. 

There is one standard caveat: following prior works
\cite{smith2020flajolet, choi2020differentially}
our privacy analysis assumes 
a perfectly random hash function. 
One can 
remove this assumption both in theory and practice by
using a cryptographic hash function. This will
yield a sketch that satisfies either a
computational variant of differential privacy called
SIM-CDP, or standard information-theoretic notions of differential privacy
under the assumption that the hash function
fools space-bounded computations \cite[Section 2.3]{smith2020flajolet}.

Other works also consider the privacy-preserving properties of common $L_p$ functions over data streams.
For $p=2$, these include fast dimensionality reduction \cite{blocki2012johnson,upadhyay2014randomness}
and least squares regression \cite{sheffet2017differentially}.
Meanwhile, for $0 < p \le 1$, frequency-moment estimation has also been studied \cite{wang2022differentially}.
Our focus is solely the cardinality estimation problem when $p=0$.

\subsection{Preliminaries}
More formally, we consider the following problem.
\paragraph{Problem Definition}
Let $\dcal = \{x_1, \ldots, x_n\}$ denote a stream of samples with each identifier $x_i$ coming from a 
large universe $U$, e.g., of size $2^{64}$.
The objective is to estimate the cardinality, or number of distinct identifiers, of $\dcal$ using an algorithm $S$ which is given privacy parameters
$\epsilon, \delta \ge 0$ and a space bound $b$, measured in bits.

\begin{definition}[Differential Privacy \cite{dwork2006calibrating}]
A randomized algorithm $S$ is $(\epsilon, \delta)$-differentially private
($(\epsilon, \delta)$-DP for short or if $\delta=0$,  \emph{pure} $\epsilon$-DP) if for any pair of data sets $\dcal, \dcal'$
that differ in one record and for all $\scal$ in the range of $S$,
$\Pr (S(\dcal') \in \scal) \le e^{\epsilon}  \Pr (S(\dcal) \in \scal) + \delta$
with probability over the internal randomness of the algorithm $S$. 
\end{definition}

Rather than analyzing any specific sketching algorithm, 
we analyze a natural class of randomized distinct counting sketches.
Algorithms in this class operate in the following manner: 
each time a new stream item $i$ arrives, 
$i$ is hashed using some uniform
random hash function $h$, and then $h(i)$
is used to update the sketch, i.e., the update procedure
depends only on $h(i)$, and is otherwise independent of $i$.
Our analysis applies to any such 
algorithm that depends only on the 
\emph{set} of observed hash values. Equivalently,
the sketch state is 
invariant both to the order in which stream items arrive,
and to item duplication.\footnote{
A sketch is \emph{duplication-invariant} if and only if its state when run on any stream
$\sigma$ is identical to its state when run on the stream $\sigma'$, in which
all elements of the stream $\sigma$ appear exactly once.
}
We call this class of algorithms \emph{hash-based, order-invariant} cardinality estimators.
Note that for any hash-based, order-invariant cardinality estimator,
the distribution of the sketch 
depends only on the cardinality of the stream.
All distinct counting sketches of which we are aware that are invariant to permutations of the input data are included in this class. 
This includes \acrshort{fm}, LPCA, Bottom-$k$, Adaptive Sampling, and HLL as shown in Section \ref{sec:example-sketches}.

\begin{definition}[Hash-Based, Order-invariant Cardinality Estimators]
Any sketching algorithm that depends only on the \emph{set} of hash values of stream items using a uniform random hash function is a \emph{hash-based  order-invariant cardinality estimator}.
We denote this class of algorithms by $\ccal$.
\end{definition}

We denote a sketching algorithm with internal randomness $r$ by $S_r$ (for hash-based algorithms, $r$ specifies the random hash function used). 
The algorithm takes a data set $\dcal$ and generates a data structure $S_r(\dcal)$ that is used to estimate the cardinality. We refer to this structure as the \emph{state of the sketch}, or simply the \emph{sketch}, and the values it can take by $s \in \Omega$.
Sketches are first initialized and then items are inserted into the sketch with an \texttt{add} operation that may or 
may not change the sketch state. 

The size of the sketch is a crucial constraint, and we denote the space consumption in bits by $b$.
For example, \acrshort{fm} consists of $k$ bitmaps of  length $\ell$. 
Thus, its state $s \in \Omega = \{0, 1\}^{k \times \ell}.$
Typically, $\ell = 32$, so that $b = 32 k$. 
Further examples are given in Section \ref{sec:example-sketches}.
Our goal is to prove such sketches are differentially private.

\label{s:fmdeets}

\section{Hash-Based Order-Invariant Estimators are Private}
\label{sec:sketches are private}

The distribution of any hash-based, order-invariant cardinality estimator 
depends only on the cardinality 
of the input stream, so without loss of generality
we assume the input is $\dcal=\{1, \dots, n\}$.
Denote the set $\dcal \backslash \{i\}$ by  $\dcal_{-i}$
for $i \in \dcal$ and 
a sketching algorithm with internal randomness $r$ by $S_r(\dcal)$.

By definition, for an $\epsilon$-differential privacy guarantee, we must show that the Bayes factor comparing the hypothesis $i \in \dcal$ versus $i \notin \dcal$ is appropriately bounded:
\begin{align}
\label{eqn:DP inequality}
    e^{-\epsilon} < \frac{\Pr_r(S_r(\dcal) = s)}{\Pr_r(S_r(\dcal_{-i}) = s)} < e^{\epsilon} \quad \forall s \in \Omega, i\in \dcal.
\end{align}

\textbf{Overview of privacy results.} 
The main result in our analysis bounds
the privacy loss of a hash-based, order-invariant sketch in terms of
just two sketch-specific quantities.
Both quantities intuitively capture how sensitive
the sketch is to the removal or insertion of a single
item from the data stream.

The first quantity is a bound $k_{max}$
on the number of items that would change the sketch if \emph{removed} from the stream. 
Denote the items whose removal from the data set changes the sketch by 
\begin{align} \label{eq:krdef}
  \kcal_r &\defn \{i \in \dcal : S_r(\dcal_{-i}) \neq S_r(\dcal) \}.
\end{align}
Denote its cardinality by $K_r \defn \left| \kcal_r \right|$
and the upper bound by $k_{max} = \sup_r K_r$. 

The second quantity is a bound on a "sampling" probability.
Let $\pi(s)$ be the probability that a newly \emph{inserted} item would change a sketch in state $s$,
\begin{equation}
    \label{eq:pi-defn}
    \pi(s) \coloneqq \Pr_r(S_r(\dcal) \neq S_r(\dcal_{-i})\, |\, S_r(\dcal_{-i}) = s).
\end{equation}
Although a sketch generally does not store explicit samples, 
conceptually, it can be helpful to think of $\pi(s)$ as the probability that an as-yet-unseen item $i$ gets ``sampled'' by a sketch in state $s$. 
We upper bound $\pi^* \coloneqq \sup_{s \in \Omega} \pi(s)$ to limit the influence of items added to the stream.

The main sub-result in our analysis
(Theorem \ref{thm:main result})
roughly states that the sketch is $\epsilon$-DP so long as 
(a) the sampling probability 
$\pi^* < 1-e^{-\epsilon}$ is small enough,
and (b) the stream cardinality $n > \frac{k_{max}}{e^{\epsilon}-1} = \Theta(k_{max}/\epsilon)$ is large enough.

We show Property (a) is a \emph{necessary} condition for any $\epsilon$-DP algorithm if the algorithm
works over data universes of unbounded size. 
Unfortunately, Property (a) does \emph{not} directly hold for natural 
sketching algorithms. 
But we show (Section \ref{s:actualdpalgs}) by applying a simple 
down-sampling procedure,
any hash-based, order-invariant 
algorithm can be modified to satisfy (a).

Furthermore, Section \ref{sec:example-sketches} shows
common sketches satisfy
Property (a) with high probability, thus providing $(\epsilon, \delta)$-DP guarantees for sufficiently large cardinalities. 
Compared to \cite{choi2020differentially}, these guarantees are tighter, more precise, and more general as they establish  the failure probability $\delta$ decays exponentially with $n$, provide explicit formulas for $\delta$, and apply to a range of sketches rather than just HLL.

\paragraph{Overview of the analysis.}
The definition of $\epsilon$-DP requires bounding the Bayes factor
in equation \ref{eqn:DP inequality}. 
The challenge is that the numerator and denominator may not be easy to compute by themselves. However, it is similar to the form of a conditional probability involving only one insertion. 
Our main trick  re-expresses this Bayes factor as a sum of conditional probabilities involving a single insertion. 
Since the denominator $\Pr_r(S_r(\dcal_{-i})=s)$ involves a specific item $i$ which may change the sketch, 
we instead consider the smallest item $J_r$ whose removal does not change the sketch. 
This allows us to re-express the numerator in terms of a conditional probability $\Pr_r(S(\dcal) = s \land J_r=j) = \Pr_r(J_r = j | S(\dcal_{-j})=s) \Pr_r( S(\dcal_{-j})=s)$
involving only a single insertion plus a nuisance term $\Pr_r( S(\dcal_{-j})=s)$. 
The symmetry of items gives that the nuisance term is equal to denominator 
$\Pr_r( S(\dcal_{-j})=s) = \Pr_r( S(\dcal_{-i})=s)$, thus allowing us to eliminate it.
\begin{lemma}
\label{lem:sketch-state}
Suppose $n > \sup_r K_r$. Then
 $\Pr_r(K_r = n) = 0$, and
\begin{equation}
\frac{\Pr_r(S_r(\dcal) = s)}{\Pr_r(S_r(\dcal_{-i}) = s)}  = 
\sum_{j \in \dcal} \Pr_r(J_r = j \,|\, S_r(\dcal_{-j}) = s). 
\label{eq:sketch-state}
\end{equation}
\end{lemma}

By further conditioning on the total number of items that, when removed, can change the sketch,
we obtain conditional probabilities that are simple to calculate. 
A combinatorial argument simplifies the resulting expression and gives us two factors in Lemma \ref{lem:sum-sampling-prob}, one involving the sampling probability for new items $\pi(s)$ given a sketch in state $s$ and the other being an expectation involving $K_r$. 
This identifies the two quantities that must be controlled in order for a sketch to be $\epsilon$-DP.
\begin{lemma}
\label{lem:sum-sampling-prob}
Under the same assumptions as Lemma \ref{lem:sketch-state} 
\begin{equation}
\sum_{j \in \dcal} \Pr_r(J_r = j \,|\, S_r(\dcal_{-j}) = s) = 
(1-\pi(s)) \E_r \left(1 + \frac{ K_r }{n -K_r+1} \bigg| S_r(\dcal_{-1}) = s \right).
\label{eq:sum-sampling-prob-main}
\end{equation}
\end{lemma}

To show that all hash-based, order invariant sketching algorithms can be made $\epsilon$-DP, we show that $K_r$ can always be bounded by the maximum size of the sketch in bits. Thus, if a sketch is combined with a downsampling procedure to ensure $\pi(s)$ is sufficiently small, one satisfies both of the properties that are sufficient for an $\epsilon$-DP guarantee.

Having established \eqref{eq:sum-sampling-prob-main}, we can derive a result showing that a hash-based, order-invariant sketch is $\epsilon$-DP
so long as the stream cardinality is large enough and $\sup_{s \in \Omega} \pi(s)$ is not too close to $1$. 

\begin{corollary}
Let $\Omega$ denote the set of all possible states of a
hash-based order-invariant distinct counting sketching algorithm.
When run on a stream of cardinality $n > \sup_r K_r$,
the sketch output by the algorithm satisfies $\epsilon$-DP if
\begin{align}
    &\pi_0 := 1-e^{-\epsilon} > \sup_{s \in \Omega} \pi(s) \quad \text{and} 
    \label{eq:downsampling probability}  \\
    &e^{\epsilon} > 1+\E_r \left(\frac{ K_r }{n -K_r+1} \bigg| S_r(\dcal_{-1}) = s \right)
 \quad \text{for all sketch states $s \in \Omega$.}
 \label{eq:raw cardinality condition}
\end{align}
Furthermore, if the data stream $\dcal$ consists of items from a universe $U$ of unbounded size, Condition \ref{eq:downsampling probability} is necessarily satisfied by \emph{any} sketching algorithm satisfying $\epsilon$-DP.
\label{cor:raw main result}
\end{corollary}

The above corollary may be difficult to apply directly since the expectation in Condition \eqref{eq:raw cardinality condition} is often difficult to compute and depends on the unknown cardinality $n$. Our main result provides sufficient criteria to ensure that Condition \eqref{eq:raw cardinality condition}
holds. The criteria is expressed in terms of 
a minimum cardinality $n_0$ and sketch-dependent constant $k_{max}$.
This constant $k_{max}$ is a bound on the maximum number of items which change the sketch when removed. 
That is, for all input streams $\dcal$ and all $r$, $k_{max} \geq |\kcal_r|$. 
We derive $k_{max}$ for a number of popular sketch algorithms in Section \ref{sec:example-sketches}.

\begin{theorem}
Consider any hash-based, order-invariant distinct counting sketch.
The sketch output by the algorithm satisfies an $\epsilon$-DP guarantee if 
\begin{align}
    &\sup_{s \in \Omega} \pi(s) < \pi_0 := 1-e^{-\epsilon}
    \quad \text{and there are strictly greater than}
    \label{eq:downsampling probability2}  \\ 
    &n_0 \coloneqq {k_{max}} /{(1-e^{-\epsilon})}
    \quad \text{unique items in the stream.} 
    \label{eq:artificial items bound}
\end{align}

\label{thm:main result}
\end{theorem}

Later, we explain how to modify existing sketching algorithms
in a black-box way to satisfy these conditions. 
If left unmodified, most sketching algorithms used in practice allow for some sketch values $s \in \Omega$
which violate Condition \ref{eq:downsampling probability2}, i.e $\pi(s) > 1-e^{-\epsilon}$. 
We call such sketch values ``privacy-violating''. Fortunately,
such values turn out to arise with only tiny probability. 
The next theorem states that, so long as this probability is smaller than $\delta$,
the sketch satisfies $(\epsilon, \delta)$-DP without modification.
The proof of Theorem \ref{thm:approxdp} follows immediately from Theorem \ref{thm:main result}.

\begin{theorem}
Let $n_0$ be as in Theorem \ref{thm:main result}. Given a hash-based, order-invariant distinct counting sketch with bounded size,
let $\Omega'$ be the set of sketch states such that $\pi(s) \geq \pi_0$.
If the input stream $\dcal$ has cardinality $n > n_0,$ 
then the sketch is $(\epsilon, \delta)$ differentially private where $\delta = \Pr_r(S_r(\dcal) \in \Omega')$. 
\label{thm:approxdp}
\end{theorem}

\subsection{Constructing Sketches Satisfying Approximate Differential Privacy: Algorithm \ref{alg:basic}}
\label{sec:approx-dp-results}
Theorem \ref{thm:approxdp} states that, when run 
on a stream with $n \geq n_0$ distinct items,
any hash-based order-invariant algorithm 
(see Algorithm \ref{alg:basic})
automatically 
satisfies $(\epsilon, \delta)$-differential privacy
where $\delta$ denotes the probability that the final
sketch state $s$ is ``privacy-violating'', i.e., $\pi(s) > \pi_0 = 1 - e^{-\epsilon}$. 
In Section \ref{sec:example-sketches}, we provide concrete bounds of $\delta$ for specific algorithms. In all cases considered, $\delta$ falls exponentially with respect to the cardinality $n$. Thus, high privacy is achieved with high probability so long as the stream is large.

We now outline how to derive a bound for a specific sketch.
We can prove the desired bound on $\delta$ by analyzing sketches in a manner similar to the coupon collector problem. Assuming a perfect, random hash function, the hash values of a universe of items defines a probability space. 
We can identify $v \le k_{max}$ events or coupons, $C_1, \ldots, C_{v}$, such that $\pi(s)$ is guaranteed to be less than $\pi_0$ after all events have occurred. Thus, if all coupons are collected, the sketch satisfies the requirement to be $\epsilon$-DP.
As the cardinality $n$ grows, the probability that a particular coupons remains missing decreases exponentially. A simple union bound shows that the probability $\delta$ that \emph{any} coupon is missing decreases exponentially with $n$.

For more intuition as to
why unmodified sketches satisfy an $(\epsilon, \delta)$-DP guarantee when the cardinality is large, we note that the inclusion probability $\pi(s)$ is closely tied to the cardinality estimate in most sketching algorithms.
For example, the cardinality estimators used in HLL and KMV are inversely proportional to the sampling
probability $\pi(s)$, i.e., $\hat{N}(s) \propto 1/\pi(s)$, while for LPCA and Adaptive Sampling, the cardinality estimators are monotonically decreasing with respect to $\pi(s)$. 
Thus, for most sketching algorithms, when run on a stream
of sufficiently large cardinality, the resulting sketch is privacy-violating only when the cardinality estimate is also inaccurate. 
Theorem \ref{thm:privacy violation}
is useful when analyzing the privacy of such algorithms,
as it characterizes the probability $\delta$ of a ``privacy
violation'' in terms of the probability the 
returned estimate, $\hat{N}(S_r(\dcal))$, is lower than some threshold
$\tilde{N}(\pi_0)$. 

\begin{theorem}
\label{thm:privacy violation}
Let $S_r$ be a sketching algorithm with estimator $\hat{N}(S_r)$. 
If $n \geq n_0$ and the estimate returned on sketch $s$ is a strictly decreasing function of $\pi(s)$, 
so that $\hat{N}(s) = \tilde{N}(\pi(s))$ for a function $\tilde{N}$.
Then, $S_r$ is $(\epsilon, \delta)$-DP where $\delta = \Pr_r(\hat{N}(S_r(\dcal)) < \tilde{N}(\pi_0))$.
\label{thm:eps-delta using Nhat}
\end{theorem}

\subsection{Constructing Sketches Satisfying Pure Differential Privacy: Algorithm \ref{alg:dp-large-set} - \ref{alg:dp-any-set}}
Theorem \ref{thm:main result} guarantees an $\epsilon$-DP sketch if \eqref{eq:downsampling probability2}, \eqref{eq:artificial items bound} hold.
Condition \eqref{eq:downsampling probability2} requires that 
$\sup_{s \in \Omega}\pi(s) < 1-e^{-\epsilon}$, 
i.e., the ``sampling probability'' of the sketching algorithm is sufficiently small regardless
of the sketch's state $s$. 
Meanwhile, \eqref{eq:artificial items bound} requires that the input cardinality is sufficiently large. 

We show that \emph{any} hash-based, order-invariant distinct counting sketching algorithm
can satisfy these two conditions by adding a simple pre-processing step which does two things. 
First, it ``downsamples'' the input stream by hashing each input,
interpreting  the hash values as numbers in $[0, 1]$, and simply
ignoring numbers whose hashes are larger than $\pi_0$. 
The downsampling hash must be independent to that used by the sketching algorithm itself. 
This 
ensures that Condition \eqref{eq:downsampling probability2} is satisfied,
as each input item has maximum sampling probability $\pi_0$. 

If there is an a priori guarantee that the number of distinct items 
$n$ is greater than $n_0 = \frac{k_{max}}{1-e^{-\epsilon}}$, 
then \eqref{eq:artificial items bound} is trivially satisfied. 
Pseudocode for the resulting $\epsilon$-DP algorithm is given in Algorithm \ref{alg:dp-large-set}. 
If there is no such guarantee, then the preprocessing step adds $n_0$ items to the input stream to satisfy
\eqref{eq:artificial items bound}.
To ensure unbiasedness, these $n_0$ items must 
(i) be distinct from any items in the ``real'' stream, and
(ii) be downsampled as per the first modification.
An unbiased estimate of the cardinality of the unmodified stream can then be easily recovered from the sketch via a post-processing correction.
Pseudocode for the modified algorithm, which is guaranteed to satisfy $\epsilon$-DP,
is given in Algorithm \ref{alg:dp-any-set}.

\label{s:actualdpalgs}

\begin{figure}
\renewcommand\figurename{Algorithms}
\hspace{-0.5cm}
    \begin{subfigure}{.29\textwidth}
        \begin{algorithmic}
            \Function{Base}{items, $\epsilon$}
            \State $S \gets InitSketch()$
            \State
            \For{$x \in items$}
            \State
            \State $S.add(x)$
            \EndFor
            \State \Return $\hat{N}(S)$
            \EndFunction
        \end{algorithmic}
    \caption{$(\epsilon, \delta)$-DP for $n \ge n_0$.}
    \label{alg:basic}
    \end{subfigure}
    \begin{subfigure}{.36
    \textwidth}
        \begin{algorithmic}
            \Function{DPSketchLargeSet}{items, $\epsilon$}
            \State $S \gets InitSketch()$
            \State $\pi_0 \gets 1-e^{-\epsilon}$
            \For{$x \in items$}
            \If{$hash(x) < \pi_0$}
            \State $S.add(x)$
            \EndIf
            \EndFor
            \State \Return $\hat{N}(S) / \pi_0$
            \EndFunction
        \end{algorithmic}
        \caption{$(\epsilon, 0)$-DP for $n \ge n_0$.}
        \label{alg:dp-large-set}
    \end{subfigure}
    \begin{subfigure}{.36\textwidth}
        \begin{algorithmic}
            \Function{DPSketchAnySet}{items, $\epsilon$}
            \State $S, n_0 \gets DPInitSketch(\epsilon)$
            \State $\pi_0 \gets 1-e^{-\epsilon}$
            \For{$x \in items$}
            \If{$hash(x) < \pi_0$}
            \State $S.add(x)$
            \EndIf
            \EndFor
            \State \Return $\hat{N}(S) / \pi_0 - n_0$
            \EndFunction
        \end{algorithmic}
        \caption{$(\epsilon, 0)$-DP for $n \ge 1$.}
        \label{alg:dp-any-set}
    \end{subfigure}%
\caption{Differentially private cardinality estimation algorithms from black box sketches.
The function $InitSketch()$ initializes a black-box sketch. 
The uniform random hash function $hash(x)$ is chosen independently of any hash in 
the black-box sketch and is interpreted as a real in $[0, 1]$.
The cardinality estimate returned by sketch $S$ is denoted $\hat{N}(S)$.
\texttt{DPInitSketch} is given in Algorithm \ref{alg:DPInitSketch}.}
\label{alg:dp sketches}
\addtocounter{algorithm}{1} 
\end{figure}

\begin{corollary} \label{maincor}
The functions \texttt{DPSketchLargeSet} (Algorithm \ref{alg:dp-large-set}) and \texttt{DPSketchAnySet} 
(Algorithm \ref{alg:dp-any-set}) yield $\epsilon$-DP distinct counting sketches provided that 
$n \ge n_0$ and $n \ge 1$, respectively.
\label{cor:algo-1b-1c-DP}
\end{corollary}

\subsection{Constructing $\epsilon$-DP Sketches from Existing Sketches: Algorithm \ref{alg:make dp}, Appendix \ref{app: algos}}
\label{sec:existing sketches post processing}
As regulations change and new ones are added, 
existing data may need to be appropriately anonymized. However, if the data has already been sketched, the underlying data may no longer be available, and even if it is retained, it may be too costly to reprocess it all. Our theory allows these sketches to be directly converted into differentially private sketches when the sketch has a merge procedure. 
Using the merge procedure to achieve $\epsilon$-differential privacy yields more useful estimates than the naive approach of simply adding Laplace noise to cardinality estimates in proportion to the global sensitivity.

The algorithm 
assumes it is possible to take a sketch $S_r(\dcal_1)$ of 
a stream $\dcal_1$ and a sketch $S_r(\dcal_2)$ of a stream $\dcal_2$,
and ``merge'' them to get a sketch of the concatenation of the two streams
$\dcal_1 \circ \dcal_2$. 
This is the case for most practical hash-based order-invariant distinct count sketches.
Denote the merge of sketches $S_r(\dcal_1)$ and $S_r(\dcal_2)$ by $S_r(\dcal_1) \cup S_r(\dcal_2)$.
In this setting, we think of the existing non-private sketch $S_r(\dcal_1)$ being converted to a sketch 
that satisfies $\epsilon$-DP by Algorithm  \ref{alg:make dp} (see pseudocode
in Appendix \ref{app: algos}).
Since sketch $S_r(\dcal_1)$ is already constructed, items cannot be first downsampled in the build
phase the way they are in Algorithms \ref{alg:dp-large-set}-\ref{alg:dp-any-set}.
To achieve $\epsilon$-DP, Algorithm \ref{alg:make dp} constructs a noisily initialized sketch, 
$S_r(\dcal_2)$, which satisfies both the downsampling condition (Condition \eqref{eq:downsampling probability2}) and the minimum 
stream cardinality requirement (Condition \eqref{eq:artificial items bound}) and returns the merged sketch 
$S_r(\dcal_1) \cup S_r(\dcal_2)$. 
Hence, the sketch will satisfy both conditions for $\epsilon$-DP, as shown in Corollary \ref{cor:makedp}

This merge based procedure typically adds no additional error to the estimates for large cardinalities. In contrast, the naive approach of adding Laplace noise can add significant noise since the sensitivity can be very large. For example, HLL's estimator is of the form $\hat{N}_{HLL}(s) = \alpha / \pi(s)$ where $\alpha$ is a constant and $s$ is the sketch. 
One item can update a bin to the maximum value, so that the updated sketch $s'$ has sampling probability $\pi(s') < \pi(s)(1-1/k)$. The sensitivity of cardinality estimate is thus at least
$\hat{N}_{HLL}(s) / k$. Given that the cardinality estimate, and hence sensitivity, can be arbitrarily large when $n \geq k$, the naive approach is unworkable to achieve $\epsilon$-DP.

\floatname{algorithm}{Algorithm}

\section{The Utility of Private Sketches}
\label{sec:utility analysis}

When processing a data set with $n$ unique items, 
denote the 
expectation and variance of a sketch and its estimator by $\E_n(\hat{N})$ and $\Var_n(\hat{N})$ respectively.
We show that our algorithms all yield unbiased estimates. Furthermore, we show that for Algorithms \ref{alg:basic}-\ref{alg:dp-any-set}, if the base sketch satisfies a \emph{relative error guarantee} (defined below), the DP sketches add no additional error asymptotically.

\textbf{Establishing unbiasedness.}
To analyze the expectation and variance of each algorithm's estimator, $\hat{N}(S(\dcal))$,
note that each estimator 
uses a `base estimate' $\hat{N}_{base}$ from the base sketch $S$ and has the form 
$\hat{N}(S(\dcal)) = \frac{\hat{N}_{base}}{p} - V$;
$p$ is the downsampling probability and $V$ is the number of artificial items added.
This allows us to express expectations and variance via the variance of the base estimator.

\begin{theorem}
Consider a base sketching algorithm $S \in \ccal$ with an unbiased estimator $\hat{N}_{base}$ for the cardinality of items added to the base sketch.
Algorithms \ref{alg:dp sketches} (a)-(c) and \ref{alg:make dp} yield unbiased estimators. 
\label{thm:unbiased}
\end{theorem}

\paragraph{Bounding the variance.}
Theorem \ref{thm:unbiased} yields a clean expression for the variance of our private algorithms.
Namely, 
$\Var [ \hat{N}(S_r(\dcal)) ] = \E [ \Var (\frac{\hat{N}_{base}}{p} | V ) ]$
which is shown in Corollary \ref{cor:var}.
The expression is a consequence of the law of total variance and that the estimators are unbiased.

We say that the base sketch satisfies a \textbf{relative-error guarantee} if with high probability, the estimate returned 
by the sketching algorithm when run on a stream of cardinality
$n$ is $(1 \pm 1/\sqrt{c}) n$ for some constant $c>0$.
Let $\hat{N}_{base,n}$ denote the  cardinality estimate when
the base algorithm is run on a stream of cardinality $n$, as opposed to $\hat{N}_{base}$ denoting the cardinality estimate produced by the base sketch on the sub-sampled stream used in our private sketches \texttt{DPSketchLargeSet} (Algorithm \ref{alg:dp-large-set}) and \texttt{DPSketchAnySet} 
(Algorithm \ref{alg:dp-any-set}).
The relative error guarantee is satisfied when $\Var_n(\hat{N}_{base,n}) < n^2 / c$; this is an immediate consequence of Chebyshev's inequality. 

When the number of artificially added items $V$ is constant as in Algorithms \ref{alg:dp-large-set} and \ref{alg:dp-any-set}, 
Corollary \ref{cor:var} provides a precise expression for the variance of the differentially private sketch. 
In Theorem \ref{thm:utility given relative error} below, we
use this expression to establish that the modification of the base algorithm to an $\epsilon$-DP 
sketch as per
Algorithms \ref{alg:dp-large-set} and \ref{alg:dp-any-set} satisfy the exact same relative error guarantee asymptotically. 
In other words, the additional error due to
any pre-processing (down-sampling and possibly adding artificial items) is insignificant for large cardinalities $n$.
 
\begin{theorem}
Suppose $\hat{N}_{base,n}$ satisfies a relative error guarantee,
$\Var_n(\hat{N}_{base,n}) < n^2 / c$, for all $n$ and for some constant $c$. 
Let $v = 0$ for Algorithm \ref{alg:dp-large-set} and $v = n_0$ for Algorithm \ref{alg:dp-any-set}.
Then Algorithms \ref{alg:dp-large-set} and \ref{alg:dp-any-set}  satisfy
\begin{align}
&\Var_n(\hat{N}) \leq \frac{(n+v)^2}{c} + \frac{(n+v)(v + \pi_0^{-1})}{k_{max}} = \frac{(n+v)^2}{c} + O(n),
\label{eq: variance result}
\end{align}
so that ${\Var_n(\hat{N})}/{\Var_n(\hat{N}_{base,n})} \to 1 \,\,\mbox{ as } n \to \infty.$
\label{thm:utility given relative error}
\end{theorem}

In Corollary \ref{cor: variance-ratios} we prove an analagous result for Algorithm \ref{alg:make dp}, 
which merges non-private and noisy sketches to produce a private sketch.
Informally, the result is comparable to \eqref{eq: variance result}, albeit with $v \ge n_0$.
This is because, in Algorithm \ref{alg:make dp}, the number of artificial items added $V$ is a random variable. 
We ensure that the algorithm satisfies a utility guarantee by bounding $V$ with high probability. 
This is equivalent to showing that the base sketching algorithm
satisfies an $(\epsilon, \delta)$-DP guarantee as
for any $n^* \geq n_0$ and dataset $\dcal^*$ with $|\dcal^*| = n^*$, $(\epsilon, \delta_{n^*})$-DP ensures
$\delta_{n^*} > \Pr_r(\pi(\scal_r(\dcal^*)) > \pi_0) = \Pr_r(V > n^*)$ 
which follows from the definition of $V$ in Algorithm \ref{alg:DPInitSketchForMerge}.

\section{Examples of Hash-based, Order-Invariant Cardinality Estimators}
\label{sec:example-sketches}

We now provide $(\epsilon, \delta)$-DP results for a select group of samples: \acrshort{fm}, LPCA, Bottom-$k$, 
Adaptive Sampling, and HLL.
The $(\epsilon, \delta)$-DP results in this section operate in the Algorithm \ref{alg:basic} setting with no 
modification to the base sketching algorithm.
Recall that the quantities of interest are the number of bins used in the sketch $k$, the size of the sketch in bits 
$b$ and the number of items whose absence changes the sketch $k_{max}$.
From Section \ref{sec:sketches are private} and Lemma \ref{lem:k bound} we know that $k_{max} \le b$
but for several common sketches 
we show a stronger bound of $k_{max} = k$.
The relationship between these parameters for various sketching algorithms is summarized in Table \ref{tab:sketch comparison}.
Table \ref{tab: related-work-bounds}, Appendix \ref{app: sketch examples}, 
details our improvements over \cite{smith2020flajolet,choi2020differentially} in both privacy and utility.

We remind the reader that, per \eqref{eq:downsampling probability},
$\pi_0 = 1-e^{-\epsilon},$
and \eqref{eq:artificial items bound} $n_0 = \frac{k_{max}}{1-e^{-\epsilon}}.$
Furthermore, recall that once we bound the parameter $k_{max}$ for any given
hash-based order-invariant sketching algorithm, 
Corollary \ref{maincor} states that the derived algorithms \ref{alg:dp-large-set}-\ref{alg:dp-any-set}
satisfy $\epsilon$-DP provided that $n \ge n_0$ and $n \ge 1$, respectively.
Accordingly, in the rest of this section,
we bound $k_{max}$ for each example sketch of interest,
which has the consequences for pure $\epsilon$-differential
privacy delineated above.

\begin{table}[]
\caption{Properties of each sketch with $k$ ``buckets'' 
(see each sketch's respective section for details of what this parameter means for the sketch).
Each sketch provides an $(\epsilon, \delta)$-DP guarantee, where the column $\ln \delta$ provides an upper bound on $\ln \delta$ established in the relevant subsection of Section \ref{s:allthesketches}. 
} 
\label{tab:sketch comparison}
\centering
\begin{tabular}{llllll}
\toprule
Sketch       & $b$: size (bits) & Standard Error        & $k_{max}$ & $\ln \delta$                                    & Reference \\ \midrule
FM85         & $32k$            & $0.649n / \sqrt{k}$   & $32k$     & $-\frac{\pi_0}{2k} n + o(1)$                    & \cite{lang2017back} \\
LPCA         & $k$              & $n/\sqrt{k}$ \tablefootnote{This approximation holds for $n < k$. A better approximation of the error is  $\sqrt{k(\exp(n/k)-n/k-1)}$}  & $k$       & $-\frac{\pi_0}{\tilde{N}(\pi_0)} n + O(\log n)$ & \cite{whang1990linear}          \\
Bottom-$k$   & $64k$            & $n/\sqrt{k}$          & $k$       & $-\frac{1}{2}\frac{\pi_0}{1-\pi_0} n + o(1)$    & \cite{giroire2009order}\\
Adaptive Sampling & $k$              & $1.2 n / \sqrt{k}$    & $k$       & $-\frac{1}{2}\frac{\pi_0}{1-\pi_0} n + o(1)$    & \cite{flajolet1990adaptive} \\ 
HLL          & $5k$             & $1.04\, n / \sqrt{k}$ & $k$       & $-\frac{\pi_0}{k} n + o(1)$                     & \cite{flajolet2007hyperloglog} \\  
\bottomrule        
\end{tabular}
\end{table}

\label{s:allthesketches}

\paragraph{Flajolet-Martin '85}
The FM85 sketch, often called \emph{Probabilistic Counting with Stochastic Averaging (PCSA)}, 
consists of $k$ bitmaps $B_i$ of length $\ell$.
Each item is hashed into a bitmap and index $(B_i,G_i)$ and sets the indexed bit in the bitmap to 1.
The chosen bitmap is uniform amongst the $k$ bitmaps and the index $G_i \sim Geometric(1/2)$. 
If $\ell$ is the length of each bitmap, then the total number of bits used by the sketch is $b = k \ell$ 
and $k_{max} = k \ell$ for 
all seeds $r$.
A typical value for $\ell$ is 32 bits, as used in Table \ref{tab:sketch comparison}. 
Past work \cite{von2019rrtxfm} proposed an $\epsilon$-DP version of \acrshort{fm} using a similar subsampling idea combined with random bit flips.

\begin{theorem}
Let  $v = \lceil -\log_2 \pi_0 \rceil$ and $\tilde{\pi}_0 \defn 2^{-v} \in (\sfrac{\pi_0}{2}, \pi_0]$.
    If $n \ge n_0$, then the \acrshort{fm} sketch is $(\epsilon, \delta)$-DP with 
    $\delta \leq kv \exp\left(- \tilde{\pi}_0 \frac{n}{k} \right)$.
    \label{thm:fm85-dp}
\end{theorem}
For any $k$, \acrshort{fm} has $k_{max} \in \{32k, 64k\}$.
This is worse than all other sketches we study which have  $k_{max} = k$, so 
\acrshort{fm} needs a larger number of minimum items $n_0$ to ensure the sketch is 
$(\epsilon, \delta)$-DP.

\paragraph{LPCA}
The Linear Probabilistic Counting Algorithm (LPCA) consists of a length-$k$ bitmap.
Each item is hashed to an index and sets its bit to $1$. 
If $B$ is the number of $1$ bits, 
the LPCA cardinality estimate is $\hat{N}_{\mathtt{LPCA}} = - k \log (1- B/k) = k \log \pi(\scal_r(\dcal))$. 
Trivially, $k_{max} = k$. 

Since all bits are expected to be $1$ after processing roughly $k \log k$ distinct items, the capacity of the sketch is bounded. 
To estimate larger cardinalities, one first downsamples the distinct items with some sampling probability $p$. 
To ensure the sketch satisfies an $\epsilon$-DP guarantee, one simply ensures $p \geq \pi_0$. In this case,
our analysis shows that LPCA is differentially private with no modifications if the cardinality is sufficiently large. Otherwise,
since the estimator $\hat{N}(s)$ is a function of the sampling probability $\pi(s)$, Theorem \ref{thm:eps-delta using Nhat} provides an $(\epsilon, \delta)$ guarantee in terms of $\hat{N}$.

\begin{theorem}
Consider a LPCA sketch with $k$ bits and downsampling probability $p$. If $p < \pi_0$ and $n > \frac{k}{1-e^{-\epsilon}}$ then LPCA is $\epsilon$-DP.
Otherwise, let $b_0 = \lceil k(1-\pi_0/p) \rceil$, $\tilde{\pi}_0 = b_0/k$, and $\mu_0$ be the expected number of items inserted to fill $b_0$ bits in the sketch.
Then, LPCA is $(\epsilon, \delta)$-DP  if $n > \mu_0$ with
\begin{align}
\delta &= \Pr_r(B < b_0)
< \frac{\mu_0}{n} \exp\left(-\frac{\tilde{\pi}_0 }{\mu_0}n \right) \exp(- \tilde{\pi}_0)
\end{align}
where $B$ is the number of filled bits in the sketch. 
Furthermore,
$\mu_0 < \tilde{N}(\tilde{\pi}_0)$
where $\tilde{N}(\tilde{\pi}) = -\frac{k}{p} \log(1-\tilde{\pi})$ is the cardinality estimate of the sketch when the sampling probability is $\tilde{\pi}$.
\label{thm:lpc-dp}
\end{theorem}

\textbf{Bottom-$k$ (also known as MinCount or KMV)} 
sketches store the $k$ smallest hash values. 
Removing an item changes the sketch if and only if
1) the item's hash value is one of these $k$ and 2) it does not collide with another item's hash value.
Thus, $k_{max} = k$.
Typically, the output size of the hash function
is large enough to ensure that the collision probability is negligible, so for practical purposes $k_{max}=k$ exactly.
Since the Bottom-$k$ estimator $\hat{N}(s) = \sfrac{(k-1)}{\pi(s)}$ is a function of the update probability $\pi(s)$, 
Theorem \ref{thm:eps-delta using Nhat} gives an $(\epsilon, \delta)$-DP guarantee in terms of the cardinality estimate by coupon collecting;
Theorem \ref{thm:kmv-dp} tightens this bound on $\delta$
for a stronger $(\epsilon, \delta)$-DP guarantee.
\begin{theorem}
Consider Bottom-$k$ with $k$ minimum values. 
Given $\epsilon > 0$, let $\pi_0, n_0$ be the corresponding subsampling and minimum cardinality to ensure the modified Bottom-$k$ sketch is $(\epsilon, 0)$-DP. 
When run on streams of cardinality $n \geq n_0$, then the unmodified sketch is $(\epsilon, \delta)$-DP,
where $\delta = P(X \leq k) < \exp(-n \alpha_n)$ where $X \sim Binomial(n, \pi_0)$
and $\alpha_n = \frac{1}{2} \frac{(\pi_0 - k/n)^2}{\pi_0(1-\pi_0) + \sfrac{1}{3n^2}} \to \frac{1}{2}\frac{\pi_0}{1 - \pi_0}$ as $n \to \infty$.
\label{thm:kmv-dp}
\end{theorem}
The closely related \textbf{Adaptive Sampling} sketch has the same privacy behavior as a bottom-$k$ sketch.
Rather than storing exactly $k$ hashes, the algorithm maintains a threshold $p$ and stores up to $k$ hash 
values beneath $p$.
Once the sketch size exceeds $k$, the threshold is halved and only hashes less than $p/2$ are kept.
Since at most $k$ hashes are stored, and the sketch is modified only if one of these hashes is removed
the maximum number of items that can modify the sketch by removal is $k_{max} = k$.

\begin{corollary}
For any size $k$ and cardinality $n$, if a bottom-$k$ sketch is $(\epsilon, \delta)$-DP,
then a maximum size $k$ adaptive sampling sketch is $(\epsilon, \delta)$-DP with the
same $\epsilon$ and $\delta$.
\label{cor:ads-dp}
\end{corollary}

\paragraph{HyperLogLog (HLL)} \label{s:hllparagraph} 
hashes each item to a bin and value $(B_i, G_i)$.
Within each bin, it takes the maximum value so each bin 
is a form of Bottom-1 sketch. 
If there are $k$ bins, then $k_{max} = k$.

Our results uniformly improve upon existing DP results on the HLL sketch and its variants.
One variation of the HLL sketch achieves $\epsilon$-DP but is far slower than HLL, as it 
requires every item to be independently hashed once for each of the $k$ bins, rather than just one time
\cite{smith2020flajolet}. 
In other words, \cite{smith2020flajolet} needs $O(k)$ update time compared to $O(1)$ for our algorithms.
Another provides an $(\epsilon, \delta)$ guarantee for streams of cardinality $n \geq n_0'$, for an $n_0'$ that is larger than our $n_0$ by a factor
of roughly (at least) $8$, with $\delta$ falling exponentially with $n$ \cite{choi2020differentially}.
In contrast, for streams with cardinality $n \geq n_0$, we provide a \emph{pure} $\epsilon$-DP guarantee using Algorithms \ref{alg:dp-large-set}-\ref{alg:dp-any-set}. 
HLL also has the following $(\epsilon, \delta)$ guarantee.

\begin{theorem}
    If $n \geq n_0$, then HLL satisfies an $(\epsilon, \delta)$-DP guarantee
    where $\delta \leq k \exp(-\sfrac{\pi_0 n}{k})$ 
    \label{thm:hll-dp}
\end{theorem}

HLL's estimator is only a function of 
$\pi(s)$ for medium to large cardinalities
as it has the form $\hat{N}(s) = \tilde{N}(\pi(s))$ when $\tilde{N}(\pi(s)) > 5 k / 2$.
Thus, if $\pi_0$ is sufficiently small so that $\tilde{N}(\pi_0(s)) > 5 k / 2$, then Theorem \ref{thm:eps-delta using Nhat} can still be applied, and HLL satisfies $(\epsilon, \delta)$-DP 
with $\delta = P(\hat{N}(\scal_r(\dcal)) < \tilde{N}(\pi_0))$.

\label{s:hlls}

\section{Empirical Evaluation}

\setcounter{figure}{0} 

\begin{figure*}[b]
    \begin{adjustbox}{max width=\linewidth,center}
    \centering
    \subfloat[Wall-clock update time in seconds vs. $k$.]{{
\begin{tikzpicture}

\definecolor{darkgray176}{RGB}{176,176,176}
\definecolor{darkorange25512714}{RGB}{255,127,14}
\definecolor{gray127}{RGB}{127,127,127}
\definecolor{lightgray204}{RGB}{204,204,204}
\definecolor{steelblue31119180}{RGB}{31,119,180}

\begin{axis}[
width=10cm, height=4cm,
legend cell align={left},
legend style={
  fill opacity=0.8,
  text opacity=1,
  at={(0.75,1.25)},
  legend columns=3,
  draw=none
},
log basis x={2},
log basis y={10},
tick align=outside,
tick pos=left,
x grid style={darkgray176},
xlabel={Number of buckets: \(\displaystyle k\)},
xmajorgrids,
xmin=107.634741152475, xmax=4870.99234305115,
xmode=log,
xtick style={color=black},
y grid style={darkgray176},
ylabel={Update Time},
ymajorgrids,
ymin=5.22176264794823e-06, ymax=0.00703665332544485,
ymode=log,
ytick style={color=black}
]
\addplot [draw=gray127, fill=gray127, forget plot, only marks]
table{%
x  y
128 8.80694824218755e-06
128 7.78960449218755e-06
128 8.35793847656248e-06
128 7.81039257812509e-06
128 8.017572265625e-06
128 7.31215820312496e-06
128 7.50004199218742e-06
128 7.91765039062501e-06
128 7.87356542968748e-06
128 7.24551660156254e-06
};
\addplot [draw=gray127, fill=gray127, forget plot, only marks]
table{%
x  y
256 7.86780664062501e-06
256 7.86923925781254e-06
256 8.09863574218752e-06
256 7.36917871093745e-06
256 7.62215527343754e-06
256 7.89178710937496e-06
256 7.49817187500001e-06
256 7.70533789062498e-06
256 7.83931640624988e-06
256 7.80495117187491e-06
};
\addplot [draw=gray127, fill=gray127, forget plot, only marks]
table{%
x  y
512 7.70347070312488e-06
512 8.21701171874992e-06
512 8.48859375000019e-06
512 7.40445507812506e-06
512 7.53861914062513e-06
512 7.2826074218749e-06
512 7.57994140625001e-06
512 8.24872070312498e-06
512 7.27550292968743e-06
512 7.49670703125002e-06
};
\addplot [draw=gray127, fill=gray127, forget plot, only marks]
table{%
x  y
1024 8.80832031249995e-06
1024 7.84130664062497e-06
1024 7.4289765625e-06
1024 8.05108007812514e-06
1024 8.19460644531258e-06
1024 7.68803320312498e-06
1024 7.44967382812511e-06
1024 7.68158886718748e-06
1024 7.55247460937488e-06
1024 7.7877011718751e-06
};
\addplot [draw=gray127, fill=gray127, forget plot, only marks]
table{%
x  y
2048 8.4019999999998e-06
2048 7.7213388671876e-06
2048 7.83335546874998e-06
2048 7.52089550781244e-06
2048 7.6156132812501e-06
2048 8.25795019531242e-06
2048 7.41706640625009e-06
2048 7.46297851562504e-06
2048 7.33542968750001e-06
2048 7.62214355468733e-06
};
\addplot [draw=gray127, fill=gray127, forget plot, only marks]
table{%
x  y
4096 8.33230273437511e-06
4096 8.09208203124997e-06
4096 7.61539453125e-06
4096 8.8383544921875e-06
4096 7.32063574218752e-06
4096 7.88676269531264e-06
4096 8.0749990234376e-06
4096 7.91635644531246e-06
4096 7.71570605468752e-06
4096 8.11934472656262e-06
};
\path [draw=gray127, fill=gray127, opacity=0.25]
(axis cs:128,8.33307957915013e-06)
--(axis cs:128,7.39319815522488e-06)
--(axis cs:256,7.54398144895215e-06)
--(axis cs:512,7.28789204489651e-06)
--(axis cs:1024,7.43175835321183e-06)
--(axis cs:2048,7.36403895993152e-06)
--(axis cs:4096,7.57607216047053e-06)
--(axis cs:4096,8.40631553484206e-06)
--(axis cs:4096,8.40631553484206e-06)
--(axis cs:2048,8.07371533694344e-06)
--(axis cs:1024,8.26499399053821e-06)
--(axis cs:512,8.15923393166599e-06)
--(axis cs:256,7.96933456667281e-06)
--(axis cs:128,8.33307957915013e-06)
--cycle;

\addplot [draw=steelblue31119180, fill=steelblue31119180, forget plot, mark=*, only marks]
table{%
x  y
128 0.0001947534960937
128 0.0001916806367187
128 0.0001908865351562
128 0.0001906302382812
128 0.0001909770214843
128 0.0001884889238281
128 0.0001884423222656
128 0.0001903996982421
128 0.0001912558652343
128 0.0001892325332031
};
\addplot [draw=steelblue31119180, fill=steelblue31119180, forget plot, mark=*, only marks]
table{%
x  y
256 0.0003476945722656
256 0.0003477007792968
256 0.0003478074130859
256 0.0003461640039062
256 0.000347002305664
256 0.0003477605
256 0.0003476954296874
256 0.0003482323144531
256 0.0003464180166015
256 0.0003460392167968
};
\addplot [draw=steelblue31119180, fill=steelblue31119180, forget plot, mark=*, only marks]
table{%
x  y
512 0.0006617968720703
512 0.000658225859375
512 0.0006626290683593
512 0.0006613275849609
512 0.0006536688505859
512 0.0006589353339843
512 0.0006540365517578
512 0.0006579443603515
512 0.0006568233642578
512 0.0006545481240234
};
\addplot [draw=steelblue31119180, fill=steelblue31119180, forget plot, mark=*, only marks]
table{%
x  y
1024 0.0012779434355468
1024 0.0012759246660156
1024 0.0012724021230468
1024 0.0012834966201171
1024 0.001278306711914
1024 0.0012748594013671
1024 0.0012753291083984
1024 0.0012758959912109
1024 0.0012845457685546
1024 0.0012753216298828
};
\addplot [draw=steelblue31119180, fill=steelblue31119180, forget plot, mark=*, only marks]
table{%
x  y
2048 0.0025446583457031
2048 0.0025256901328125
2048 0.0025184860439453
2048 0.0025262350751953
2048 0.0025128205527343
2048 0.0025167101445312
2048 0.002510332977539
2048 0.0025115209082031
2048 0.0025093529609374
2048 0.0025126637207031
};
\addplot [draw=steelblue31119180, fill=steelblue31119180, forget plot, mark=*, only marks]
table{%
x  y
4096 0.0050601207509765
4096 0.0050712372246093
4096 0.0050448133769531
4096 0.005042267290039
4096 0.005030424102539
4096 0.0050595495146484
4096 0.0050312537080078
4096 0.0050349371347656
4096 0.0050319887353515
4096 0.0050670429785156
};
\path [draw=steelblue31119180, fill=steelblue31119180, opacity=0.25]
(axis cs:128,0.000192501718730927)
--(axis cs:128,0.000188847735370533)
--(axis cs:256,0.000346467276445516)
--(axis cs:512,0.000654734190455218)
--(axis cs:1024,0.00127354640401624)
--(axis cs:2048,0.00250797826172078)
--(axis cs:4096,0.00503159352421204)
--(axis cs:4096,0.00506313343906912)
--(axis cs:4096,0.00506313343906912)
--(axis cs:2048,0.00252971591074008)
--(axis cs:1024,0.00128125868719458)
--(axis cs:512,0.000661253003490022)
--(axis cs:256,0.000348035633905944)
--(axis cs:128,0.000192501718730927)
--cycle;

\addplot [draw=darkorange25512714, fill=darkorange25512714, forget plot, mark=x, only marks]
table{%
x  y
128 9.61058007812503e-06
128 8.89216503906242e-06
128 8.85341992187491e-06
128 8.36817382812499e-06
128 8.79020800781264e-06
128 9.37869140624995e-06
128 8.49498242187492e-06
128 8.78404394531263e-06
128 9.02319042968736e-06
128 8.7716279296875e-06
};
\addplot [draw=darkorange25512714, fill=darkorange25512714, forget plot, mark=x, only marks]
table{%
x  y
256 9.04831250000005e-06
256 8.53144921875008e-06
256 8.37352246093746e-06
256 8.79887207031263e-06
256 8.66394824218736e-06
256 8.48463867187516e-06
256 8.49174121093747e-06
256 8.97553124999998e-06
256 8.48565136718753e-06
256 8.9491054687499e-06
};
\addplot [draw=darkorange25512714, fill=darkorange25512714, forget plot, mark=x, only marks]
table{%
x  y
512 8.77659667968748e-06
512 9.2303544921876e-06
512 8.14923925781239e-06
512 8.40567773437487e-06
512 8.66713574218764e-06
512 9.38669921874986e-06
512 8.67922558593745e-06
512 8.23033691406232e-06
512 9.09257714843759e-06
512 8.28521386718743e-06
};
\addplot [draw=darkorange25512714, fill=darkorange25512714, forget plot, mark=x, only marks]
table{%
x  y
1024 8.68928515625005e-06
1024 8.93816601562486e-06
1024 8.50683984375002e-06
1024 8.41088867187495e-06
1024 9.04905859374998e-06
1024 9.05335644531247e-06
1024 8.39883886718748e-06
1024 9.16827539062517e-06
1024 9.26044824218759e-06
1024 8.4463818359375e-06
};
\addplot [draw=darkorange25512714, fill=darkorange25512714, forget plot, mark=x, only marks]
table{%
x  y
2048 8.5907294921873e-06
2048 8.50962597656255e-06
2048 8.37417578125003e-06
2048 8.41953515625e-06
2048 9.74085742187509e-06
2048 8.3379218749997e-06
2048 9.32793066406269e-06
2048 8.50731640624977e-06
2048 8.31016992187527e-06
2048 8.94233203124978e-06
};
\addplot [draw=darkorange25512714, fill=darkorange25512714, forget plot, mark=x, only marks]
table{%
x  y
4096 9.21443750000021e-06
4096 8.91085644531267e-06
4096 8.53051367187508e-06
4096 8.75061523437481e-06
4096 9.59826855468759e-06
4096 8.48358398437485e-06
4096 8.94680371093777e-06
4096 9.19712500000022e-06
4096 9.47621679687474e-06
4096 9.08348144531244e-06
};
\path [draw=darkorange25512714, fill=darkorange25512714, opacity=0.25]
(axis cs:128,9.26792640402017e-06)
--(axis cs:128,8.5254901975423e-06)
--(axis cs:256,8.43597036760527e-06)
--(axis cs:512,8.25620616843936e-06)
--(axis cs:1024,8.45419966961649e-06)
--(axis cs:2048,8.22425291611676e-06)
--(axis cs:4096,8.64906893259661e-06)
--(axis cs:4096,9.38931153615346e-06)
--(axis cs:4096,9.38931153615346e-06)
--(axis cs:2048,9.18786602919568e-06)
--(axis cs:1024,9.13010814288352e-06)
--(axis cs:512,9.12440515968557e-06)
--(axis cs:256,8.92458412458226e-06)
--(axis cs:128,9.26792640402017e-06)
--cycle;

\addplot [semithick, gray127, dashed]
table {%
128 7.86313886718751e-06
256 7.75665800781248e-06
512 7.72356298828125e-06
1024 7.84837617187502e-06
2048 7.71887714843748e-06
4096 7.9911938476563e-06
};
\addlegendentry{HLL}
\addplot [semithick, steelblue31119180, dotted]
table {%
128 0.00019067472705073
256 0.00034725145517573
512 0.00065799359697262
1024 0.00127740254560541
2048 0.00251884708623043
4096 0.00504736348164058
};
\addlegendentry{QLL}
\addplot [semithick, darkorange25512714]
table {%
128 8.89670830078123e-06
256 8.68027724609376e-06
512 8.69030566406246e-06
1024 8.79215390625001e-06
2048 8.70605947265622e-06
4096 9.01919023437504e-06
};
\addlegendentry{PHLL}
\end{axis}

\end{tikzpicture}}
    \label{fig:timing-comparison}}%
    \subfloat[
    Estimated space increase using QLL \cite{smith2020flajolet} rather than PHLL.]
    {{\input{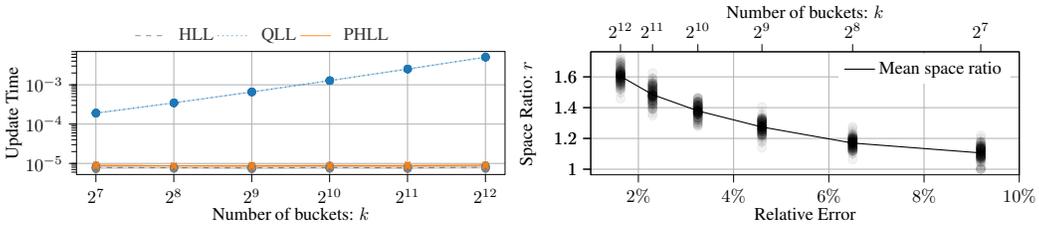}} \label{fig:space-ratio}}
    \end{adjustbox}
    \caption{
    (\ref{fig:timing-comparison}) QLL's update time is not competitive since it performs $O(k)$ hashes. (\ref{fig:space-ratio}) QLL is less efficient spacewise than PHLL. The relative size of a QLL sketch to a PHLL sketch, the \emph{Space ratio}, is larger for more accurate sketches.
    }
\end{figure*}

We provide two experiments highlighting the practical benefits of our approach.
Of past works, only \cite{choi2020differentially, smith2020flajolet} are comparable and both differ from our approach in significant ways. 
We empirically compare only to \cite{smith2020flajolet} since \cite{choi2020differentially} is simply an analysis of HLL.
Our improvement over \cite{choi2020differentially} for HLL consists of providing
significantly tighter privacy bounds in Section \ref{s:hllparagraph} and providing a fully $\epsilon$-DP sketch in the secret hash setting. We denote our $\epsilon$-DP version of HLL using Algorithm \ref{alg:dp-large-set}  by PHLL (private-HLL) and that of \cite{smith2020flajolet} by QLL.
Details of the experimental setup are in Appendix \ref{app:exptdetails}.

\textbf{Experiment 1: Update Time (Figure \ref{fig:timing-comparison}).}
We implemented regular, non-private HLL,
our PHLL, and QLL
and recorded the time to populate every sketch over $2^{10}$ updates with $k \in \{2^7, 2^8, \dots 2^{12} \}$ buckets. For HLL, these bucket sizes correspond to relative standard errors ranging from $\approx 9\%$ down to $\approx 1.6\%$.
Each marker represents the mean update time over all updates and the curves are the evaluated mean update time over $10$ trials.

As expected from theory, the update time of \cite{smith2020flajolet} grows as $O(k)$.
In contrast, our method PHLL has a constant update time and is similar in magnitude to HLL.
Both are roughly $500\times$ faster than \cite{smith2020flajolet} when $k=2^{12}$. 
Thus, figure \ref{fig:timing-comparison} 
shows that \cite{smith2020flajolet} is not a scalable solution and the speedup by achieving $O(1)$ updates is substantial.

\textbf{Experiment 2: Space Comparison (Figure \ref{fig:space-ratio}).}
In addition to having a worse update time, we also show that QLL has lower utility in the sense that it requires more space than PHLL to achieve the same error. 
Fixing the input cardinality at $n=2^{20}$ and 
the privacy budget at $\epsilon = \ln(2)$, 
we vary the number of buckets  $k \in \{2^7, 2^8, \dots 2^{12} \}$ and
simulate the $\epsilon$-DP methods, PHLL and QLL \cite{smith2020flajolet}. 
The number of buckets controls the error and we found that both methods obtained very similar mean relative error for a given number of bins\footnote{
This is shown in Figure \ref{fig:utility-space}, Appendix \ref{app:exptdetails}.} 
so we plot the space usage against the expected relative error for a given number of buckets. 
For QLL, since the error guarantees tie the parameter $\gamma$ to the number of buckets, we modify $\gamma$ accordingly as well.
We compare the sizes of each sketch as the error varies.

Since the number of bits required for each bin depends on the range of values the bin can take,
we record the simulated 
$\textbf{total sketch size} \defn k \cdot \log_2 \max_i s_i$, 
by using the space required for the largest bin value over $k$ buckets.
Although QLL achieves similar utility, it does so using a sketch that is larger: 
when $k = 2^7$, we expect an error of roughly $9\%$, QLL is roughly $1.1\times$ larger.
This increases to about $1.6\times$ larger than our PHLL sketch when $k=2^{12}$, achieving error of roughly $1.6\%$.
We see that the average increase in space when using QLL compared to PHLL 
\emph{grows exponentially in the desired accuracy of the sketch};
when lower relative error is necessary, we obtain a greater space improvement over QLL than at higher relative errors.
This supports the behavior expected by comparing with space bounds of \cite{smith2020flajolet} with (P)HLL.

\section{Conclusion} 
We have studied the (differential) privacy of a class of cardinality estimation sketches that includes most popular algorithms.
Two examples are the \acrshort{hll} and KMV (bottom-$k$) sketches that have been deployed in large systems 
\cite{heule2013hyperloglog, datasketches}.
We have shown that the sketches returned by these algorithms
are $\epsilon$-differentially private when run on streams
of cardinality greater than $n_0 = \frac{k_{max}}{1-e^{-\epsilon}}$
and when combined with a simple downsampling procedure. 
Moreover, even without downsampling, 
these algorithms satisfy $(\epsilon, \delta)$-differential
privacy where $\delta$ falls exponentially with the stream
cardinality $n$
once $n$ is larger than the threshold $n_0$.
Our results are more general and yield better 
privacy guarantees than prior work for small space cardinality estimators that preserve 
differential privacy.
Our empirical validations show that our approach is practical and scalable, being much faster than previous state-of-the-art while consuming much less space.

\begin{ack}
    We are grateful to Graham
    Cormode for valuable comments on an earlier version
    of this manuscript. Justin Thaler was supported by 
     NSF SPX award CCF-1918989 and NSF CAREER award CCF-1845125.
\end{ack}

\newpage 
\bibliography{references}

\newpage
\section{Paper Checklist}

\begin{enumerate}
\item
  \begin{enumerate}
    \item{Do the main claims made in the abstract and introduction accurately reflect the paper's contributions and scope? \answerYes}
    \item{Have you read the ethics review guidelines and ensured that your paper conforms to them? \answerYes}
    \item{Did you discuss any potential negative societal impacts of your work? 
    See answer to next question.}
    \item{Did you describe the limitations of your work? 
    Our work shows that existing algorithms, or mild variants thereof, preserve privacy.
    Therefore, there should not be any negative societal impacts that are consequences of positive privacy results
    unless users/readers incorrectly apply the results to their systems.
    Any mathematical limitations from the theory are clearly outlined through the formal technical statements.}
  \end{enumerate}
  
\item
    \begin{enumerate}
        \item{Did you state the full set of assumptions of all theoretical results?\answerYes}
        \item{Did you include complete proofs of all theoretical results?\answerYes}
    \end{enumerate}

\item
    \begin{enumerate}
        \item{Did you include the code, data, and instructions needed to reproduce the main experimental results (either in the supplemental material or as a URL)?
        A small repository containing the experimental scripts and figure plotting has been provided.}
        \item{ Did you specify all the training details (e.g., data splits, hyperparameters, how they were chosen)?
        \answerYes}
        \item{Did you report error bars (e.g., with respect to the random seed after running experiments multiple times)? 
        For Figure \ref{fig:timing-comparison} the standard deviations have been plotted in shaded regions but these
        are too small in magnitude to be seen on the scale of the plot, indicating that there is very little variation.
        For Figure \ref{fig:space-ratio} we have plotted the entire distribution over all trials.}
        \item{Did you include the amount of compute and the type of resources used (e.g., type of GPUs, internal cluster, or cloud provider)? \answerNo}
    \end{enumerate}
    
\item
    \begin{enumerate}
        \item{If your work uses existing assets, did you cite the creators? \answerNA}
        \item{Did you mention the license of the assets? \answerNA}
        \item{Did you include any new assets either in the supplemental material or as a URL? \answerNo}
        \item{Did you discuss whether and how consent was obtained from people whose data you're using/curating? \answerNA}
        \item{Did you discuss whether the data you are using/curating contains personally identifiable information or offensive content? \answerNA}
    \end{enumerate}
    
\item
    \begin{enumerate}
        \item{Did you include the full text of instructions given to participants and screenshots, if applicable? \answerNA}
        \item{Did you describe any potential participant risks, with links to Institutional Review Board (IRB) approvals, if applicable? \answerNA}
        \item{Did you include the estimated hourly wage paid to participants and the total amount spent on participant compensation? \answerNA}
    \end{enumerate}
\end{enumerate}


\newpage
\appendix

\section{Omissions from Section \ref{sec:sketches are private}}
\label{app:section3 proofs}

\subsection{Algorithms}
\label{app: algos}

Due to space considerations, the algorithms for initializing $\epsilon$-DP sketches has been omitted from the main body.
So too has the algorithm for constructing an existing non-private sketch from a private sketch.
These are detailed in Algorithms \ref{alg:init} and Algorithm \ref{alg:make dp}, respectively.

\begin{figure} 
\renewcommand\figurename{Algorithms}
\centering
    \begin{subfigure}{.45\textwidth}
        \begin{algorithmic}
        \Function{DPInitSketch}{$\epsilon$}
        \State $S \gets InitSketch()$
        \State $\pi_0 \gets 1-e^{-\epsilon}$
        \State $n_0 \gets \left\lceil \frac{k_{max}-1}{\pi_0} \right\rceil$
        \State $M \sim Binomial(n_0, \pi_0)$
        \For{$i=1 \to M$}
            \State $x \gets NewItem()$
            \If{$hash(x) < \pi_0$}
                \State $S.add(x)$
            \EndIf
        \EndFor
        \State
        \State \Return $S, n_0$
        \EndFunction
        \end{algorithmic}
    \caption{}
    \label{alg:DPInitSketch}
    \end{subfigure}
    \begin{subfigure}{.45\textwidth}
       \begin{algorithmic}
            \Function{DPInitSketchForMerge}{$\epsilon$}
            \State $S \gets InitSketch()$
            \State $\pi_0 \gets 1-e^{-\epsilon}$
            \State $n_0 \gets \left\lceil \frac{k_{max}-1}{\pi_0} \right\rceil$
            \State $v \gets 0$
            \Repeat
                \State $x \gets NewItem()$
                \State $S.add(x)$
                \State $v \gets v+1$
            \Until{$\pi(S) \leq \pi_0$ and $v \geq n_0$}
            \State \Return $S, v$
            \EndFunction
        \end{algorithmic}
        \caption{}
        \label{alg:DPInitSketchForMerge}
    \end{subfigure}%
\caption{Initialization routines for generating $\epsilon$-DP sketches.
The function $NewItem()$ returns an item that is guaranteed to come from a data universe disjoint from the universe over which stream items are drawn. In \texttt{DPInitSketch}, the binomial draw $M$ simulates inserting $n_0$ unique items into the sketch, with downsampling probability $\pi_0$.}
\label{alg:init}
\addtocounter{algorithm}{1} 
\end{figure}

\begin{algorithm}
\begin{algorithmic}
\Function{MakeDP}{$S$, $\epsilon$}
\State $T, v \gets DPInitSketchForMerge(\epsilon)$ \Comment{Algorithm \ref{alg:DPInitSketchForMerge}}
\State \Return $S \cup T,\, \hat{N}(S \cup T) - v$
\State \Comment{return private sketch and associated cardinality estimate for stream $S$ is a sketch of.}
\EndFunction
\end{algorithmic}
\caption{Turn an existing sketch into one with an $\epsilon$-DP guarantee.}
\label{alg:make dp}
\end{algorithm}

\subsection{Technical Results}

Our first result shows that all sketches under consideration cannot be modified by too many items.
Specifically, we show that the number of items that can change the sketch, $K_r$, is bounded above by 
$b$ which is the number of bits required to store the sketch.  
However, as shown in Section \ref{sec:example-sketches}, this can in fact be strengthened for many 
specific sketches.

\begin{lemma}
\label{lem:k bound}
Suppose that $n > \sup_r K_r$.
Then for any distinct counting sketch with size in bits bounded by $b$, $\sup_r K_r \leq b$. 
\end{lemma}

\begin{myproof}{Lemma}{\ref{lem:k bound}}{
Consider some data set $\dcal$ and sketch $s = \scal_r(\dcal)$.
Recall that we denote the set of items whose removal would change the sketch by
$\kcal_r(\dcal) := \{i \in \dcal : S_r(\dcal_{-i}) \neq S_r(\dcal) \}$.
Consider any subset $\Lambda \subset \kcal_r(\dcal)$.
Then we claim that, for any $x \in \kcal_r$, adding $x$ to the sketch
$S_r(\Lambda)$ will change it if and only if $x \in \kcal_r(\dcal) \setminus \Lambda$. That is, if $\Lambda \circ x$ denotes the stream
consisting of one occurrence of each item in $\Lambda$, followed by $x$,
then $S_r(\Lambda) \neq S_r(\Lambda \circ x)$ if and only if $x \in \kcal_r(\dcal) \setminus \Lambda$.

To see this, first observe that
 duplication-invariance of the sketching algorithm
 implies that if $x \in \Lambda$ 
then $S_r(\Lambda) = S_r(\Lambda \circ x)$.
Second if $x \not\in\Lambda$, 
suppose by way of contradiction that 
$S_r(\Lambda) = S_r(\Lambda \circ x)$,
and let $T=\dcal \setminus \left(\Lambda \cup \{x\}\right)$.
Since $S_r(\Lambda) = S_r(\Lambda \circ x)$, it holds that
$S_r(\Lambda \circ x \circ T) = S_r(\Lambda \circ T) = S_r(\dcal_{-x})$.
Yet by order-invariance of the sketching algorithm, 
$S_r(\Lambda \circ x \circ T) = S_r(\dcal)$.
We conclude that $S_r(\dcal_{-x}) = S_r(\dcal)$,
contradicting the assumption that $x \in \kcal_r(\dcal).$

The above means that for any fixed $r$, the sketch $S_r(\Lambda)$ 
losslessly encodes the arbitrary subset $\Lambda$
of $\kcal_r$. Hence, the sketch requires at least $\log_2(2^{|\kcal_r|}) = |\kcal_r|$ bits to represent. Thus, any sketch with size bounded by $m$ bits can have at most $m$ items that affect the sketch.
}
\end{myproof}

Next we prove Lemma \ref{lem:sketch-state} which expresses the Bayes factor from 
\eqref{eqn:DP inequality}
\begin{align*}
    \frac{\Pr_r(S_r(\dcal) = s)}{\Pr_r(S_r(\dcal_{-i}) = s)}
\end{align*}
to a sum of conditional probabilities involving a single insertion.

\begin{myproof}{Lemma}{\ref{lem:sketch-state}}{
Recall from Equation \eqref{eq:krdef} that 
\begin{align*} 
  \kcal_r &\defn \{i \in \dcal : S_r(\dcal_{-i}) \neq S_r(\dcal) \} 
\end{align*}
denotes the set of items that would change the state of the sketch if removed, 
and its cardinality is $K_r \defn \left| \kcal_r \right|$.
We also define $J_r \defn \min\{i: S_r(\dcal_{-i}) = S_r(\dcal)\}$ to be the smallest index
amongst the remaining $n-K_r$ items in $\dcal$ that do not change the sketch.
If removing \emph{any} item changes the sketch, 
i.e., if $S_r(\dcal_{-i}) \neq S_r(\dcal)$ for all $i \in \dcal$,
then $K_r = n$.
For this case, we define $J_r$ to be a special symbol $\perp$.

First, let us rewrite $\Pr_r(S_r(\dcal) = s)$ as a sum
over all possible values of $J_r$:
\begin{align} \label{eq1}
    \Pr_r(S_r(\dcal) = s) 
&= \sum_{j \in \dcal \cup \{\perp\}} \Pr_r(J_r = j \land S_r(\dcal) = s).
\end{align}
Next, we split the right hand side of Equation \eqref{eq1} into the distinct cases wherein $J_r = \perp$ and $j \in \dcal$, as we will ultimately deal with each case separately:
\begin{align} \label{eq2}
    \sum_{j \in \dcal \cup \{\perp\}} \Pr_r(J_r = j \land S_r(\dcal) = s)= \notag\\
 \Pr_r(J_r = \perp \land \, S_r(\dcal) = s) +\sum_{j \in \dcal} \Pr_r(J_r = j \land S_r(\dcal_{-j}) = s). \end{align}
 
Next, the summands over $j \in \dcal$ are decomposed via conditional probabilities.
 Specifically, the right hand side of Equation \eqref{eq2} equals:
 \begin{align}
& \Pr_r(J_r = \perp \land \, S_r(\dcal) = s) + \sum_{j \in \dcal} \Pr_r(J_r = j \,|\, S_r(\dcal_{-j}) = s) \Pr_r(S_r(\dcal_{-j}) = s) \label{eq3}.
\end{align}

For any hash-based, order-invariant sketch, the distribution of $S(\dcal)$
depends only on the number of distinct elements in $\dcal$, and hence 
the factor $\Pr_r(S_r(\dcal_{-j}) = s)$ appearing in the $j$th
summand of Equation \eqref{eq3} equals $\Pr_r(S_r(\dcal_{-i}) = s)$,
where $i$ is the element of $\dcal$ referred to in the statement of the lemma.
Accordingly, 
we can rewrite Expression \eqref{eq3} as:

\begin{align} 
\Pr_r(J_r = \perp \land \, S_r(\dcal) = s) + \sum_{j \in \dcal} 
\Pr_r(J_r = j \,|\, S_r(\dcal_{-j}) = s) \Pr_r(S_r(\dcal_{-i}) = s).  
\label{eq4}
\end{align}

Clearly, $J_r \neq \perp$ whenever 
the number of items in the data set, $n$,
exceeds $K_r$.
Hence, if $n > \sup_r K_r$,
$\Pr_r(J_r = \perp \land \, S_r(\dcal) = s)=0$. 
We obtain Equation \eqref{eq:sketch-state} as desired, namely:
\begin{equation*}
\frac{\Pr_r(S_r(\dcal) = s)}{\Pr_r(S_r(\dcal_{-i}) = s)}  = 
\sum_{j \in \dcal} \Pr_r(J_r = j \,|\, S_r(\dcal_{-j}) = s).  
\end{equation*}

}
\end{myproof}

We subsequently prove Lemma \ref{lem:sum-sampling-prob} which bounds the sum 
$\sum_{j \in \dcal} \Pr_r(J_r = j \,|\, S_r(\dcal_{-j}) = s)$ in terms of two parameters that 
we need to control.
The first is related to the ``sampling probability'' of the sketch, $\pi(s)$ and the second
is an unwieldy expectation.
Although the expectation may be difficult to compute, we will later show a more practical variant
that will be easier for us to leverage algorithmically.
For clarity, Lemma \ref{lem:sum-sampling-prob-supp} is a restatement of Lemma 
\ref{lem:sum-sampling-prob} from the main body.

\begin{lemma}
\label{lem:sum-sampling-prob-supp}
Under the same assumptions as Lemmas \ref{lem:sketch-state} and \ref{lem:k bound},
\begin{equation}
\sum_{j \in \dcal} \Pr_r(J_r = j \,|\, S_r(\dcal_{-j}) = s) = 
(1-\pi(s)) \E_r \left(1 + \frac{ K_r }{n -K_r+1} \bigg| S_r(\dcal_{-1}) = s \right).
\label{eq:sum-sampling-prob-supp}
\end{equation}
\end{lemma}

\begin{myproof}{Lemma}{\ref{lem:sum-sampling-prob}}{
We begin by writing 
\begin{align}
    \Pr_r(J_r = j \,&|\, S_r(\dcal_{-j}) = s) = \notag  \\ 
    &\sum_{k=1}^n \Pr_r(J_r = j \,|\, K_r=k, S_r(\dcal_{-j}) = s) 
    \Pr_r(K_r=k | S_r(\dcal_{-j}) = s). \label{eq5} 
\end{align} 

To analyze Expression \eqref{eq5},
 we first focus on the $\Pr_r(J_r = j \,|\, K_r=k, S_r(\dcal_{-j}) = s) $ term.
Given that $K_r = k$ and $S_r(\dcal_{-j}) = s$, we know that $J_r = j$ if and only if the items
$\{1, \dots, j - 1 \}$ are \emph{all} in $\kcal_r$ \emph{and} item $j$ is not in $\kcal_r$.
The first condition occurs with probability $\frac{ {n-(j-1) \choose k-(j-1)}}{{n \choose k}}=\frac{ {k \choose j-1} }{{n \choose j-1}}$. This is because there are ${n-(j-1) \choose k-(j-1)}$ subsets $K_r$ of $\{1, \dots, n\}$ of size $k$ that contain
items $1, \dots, j-1$, out of ${n \choose k}$ subsets of $\kcal_r$ of size $k$.
Meanwhile, $j \notin \kcal_r$ means that $S_r(\dcal_{-j}) = S_r(\dcal)$, which occurs with
probability that is
exactly the complement of the sampling probability $\pi(s)$ (see Equation \eqref{eq:pi-defn}).

By the above reasoning, the left hand side of Expression \eqref{eq:sum-sampling-prob-supp} equals:

\begin{align}
  & \sum_{j \in \dcal} \sum_{k=1}^n \Pr_r(J_r = j \,|\, K_r=k, S_r(\dcal_{-j}) = s) \Pr_r(K_r=k | S_r(\dcal_{-j}) = s) \nonumber \\
&= \sum_{k=1}^n \sum_{j \in \dcal} \frac{ {k \choose j-1} }{{n \choose j-1}} (1-\pi(s)) \Pr_r(K_r=k | S_r(\dcal_{-1}) = s) \nonumber \\
&= (1-\pi(s)) \sum_{k=1}^n \left(1 + \frac{ k }{n -k+1}\right)  \Pr_r(K_r=k | S_r(\dcal_{-1}) = s) \nonumber \\
&= (1-\pi(s)) \E_r \left(1 + \frac{ K_r }{n -K_r+1} \bigg| S_r(\dcal_{-1}) = s \right). \label{eqn:simplified LR}
\end{align}
}
\end{myproof}

By comparing the LHS of \eqref{eq:sketch-state} to the right hand side of 
\eqref{eqn:simplified LR}, overall, Lemmas \ref{lem:sketch-state} and 
\ref{lem:sum-sampling-prob} show:

\begin{equation*}
    \frac{\Pr_r(S_r(\dcal) = s)}{\Pr_r(S_r(\dcal_{-i}) = s)} = 
    (1-\pi(s)) \E_r \left(1 + \frac{ K_r }{n -K_r+1} \bigg| S_r(\dcal_{-1}) = s \right).
\end{equation*}
Hence, $\epsilon$-DP is a consequence of ensuring 
\[ (1-\pi(s)) \E_r \left(1 + \frac{ K_r }{n -K_r+1} \bigg| S_r(\dcal_{-1}) = s \right)
\in [e^{-\epsilon}, e^{\epsilon}] \]
and Corollary \ref{cor:raw main result} presents the conditions under which this is true.

\begin{myproof}{Corollary}{\ref{cor:raw main result}}{
Since $1-\pi(s) \leq 1$ for all $s$ and $\frac{ n+1 }{n -K_r+1}=1+\frac{K_r}{n -K_r+1} \geq 1$, 
it follows from Lemma \ref{lem:sum-sampling-prob}, \eqref{eq:downsampling probability} and 
\eqref{eq:raw cardinality condition}
that the right hand side of Equation \eqref{eq:sketch-state} lies in the
interval $[e^{-\epsilon}, e^{\epsilon}]$, as required for an $\epsilon$-DP guarantee.

For the necessity of Condition \ref{eq:downsampling probability}, note that
if the universe of possible items is infinite, then for any possible sketch state $s$,
there exists an arbitrarily long sequence of distinct 
items that results in state $s$ if $\pi(s) < 1$. 
One simply needs to search for a sequence of items which do not change the sketch.
Combining Lemma \ref{lem:sketch-state} with Equation \eqref{eq:sum-sampling-prob-supp} in Lemma \ref{lem:sum-sampling-prob}
therefore implies that
\begin{align}
    e^{-\epsilon} &< \inf_s (1-\pi(s)) 
    \end{align}
    and hence
    $
    \sup_s \pi(s) < 1-e^{-\epsilon}$ 
    as claimed. 
 }
\end{myproof}

As previously described, because the expected value in Lemma \ref{lem:sum-sampling-prob} depends
on the unknown cardinality $n$, it is difficult to use.  
However, we know from Lemma \ref{lem:k bound} that there are only a bounded number of items 
that actually change the sketch.
Thus, we introduce a sketch specific parameter $k_{max} \ge K_r$ which is a bound on the maximum
number of items that can change the sketch.
Although we trivially have $k_{max} \le b$ from Lemma \ref{lem:k bound}, 
Section \ref{sec:example-sketches} in fact shows that $k_max = k$, the number of ``buckets'' used
in the sketch for many popular algorithms.

\begin{myproof}{Theorem}{\ref{thm:main result}}{
We can upper bound the expectation on the right hand side of Condition \eqref{eq:raw cardinality condition} using $k_{max}$ and $n_0$. 
Corollary \ref{cor:raw main result} and solving for $n_0$ then gives the desired result.
Specifically, by Corollary \ref{cor:raw main result}, the sketch satisfies $\epsilon$-DP if:
\begin{align}
    e^{\epsilon} &> 
    \sup_{n \geq n_0} \E_r \left(1 + \frac{ K_r }{n -K_r+1} \bigg| S_r(\dcal_{-1}) = s \right) \\
    &\geq \sup_{n \geq n_0} \left(1 + \frac{k_{max}}{n-k_{max}+1}\right) \\
    &= 1 + \frac{k_{max}}{n_0-k_{max}+1} \label{eq:dp-upper}.
    \end{align}
    This is satisfied if $n_0 - k_{max}+1 > \frac{k_{max}}{e^{\epsilon}-1}$ so that
    \begin{align}
    n_0 &> k_{max} \left(1+\frac{1}{e^{\epsilon}-1}\right) -1 = \frac{k_{max}}{1-e^{-\epsilon}} - 1
    \label{eq:number-of-items-added}
    \end{align}
}
\end{myproof}
In summary, all of the technical results leading to Theorem \ref{thm:main result} are used to 
show that, provided $\sup_s \pi(s) < 1-e^{-\epsilon}$ and at least $n_0$ items appear
in the stream, then any sketching algorithm from the hash-based order-invariant class will
satisfy $\epsilon$-DP.
Theorem \ref{thm:approxdp} concludes the subsection and shows that if $\pi(s)$ is too large,
then there is a small, $\delta$ probability that the sketch is privacy violating.
In other words, if there are at least $n_0$ items in the stream, but the hash-based downsampling
rate $\pi(s)$ is not small enough, then there is a tiny $\delta$ chance the sketch may not be 
$\epsilon$-DP, and hence the sketch in this case is $(\epsilon, \delta)$-DP.

\subsection{Algorithmic Approach: Sections \ref{sec:approx-dp-results} and \ref{s:actualdpalgs}}

We separate our algorithms into three regimes described by Algorithms \ref{alg:basic}, 
\ref{alg:dp-large-set}, \ref{alg:dp-any-set}.
Algorithm \ref{alg:basic} is a ``base'' sketching algorithm chosen from the hash-based, 
order-invariant class, for example, a HyperLogLog or Bottom-$k$ sketch. 
No modifications are made to the inner workings of the algorithm but it must be implemented 
using a perfectly random hash function (see the final paragraph of Section \ref{sec:related-work}).
We show in Theorem \ref{thm:privacy violation} that on streams with at least $n_0$ items that 
the quality of the estimator is related to the downsampling probability as presented in 
Theorem \ref{thm:approxdp}.
It is used to relate the probability of a privacy-violation back to the estimate quality, rather 
than simply the state of the sketch.

\begin{myproof}{Theorem}{\ref{thm:privacy violation}}{
As $\tilde{N}$ is invertible and decreasing, $P(\hat{N}(S_r(\dcal)) < \tilde{N}(\pi_0)) = 
P(\pi(S_r(\dcal)) > \pi_0) = \delta$. 
}
\end{myproof}

The final result of Section \ref{sec:sketches are private}, Corollary \ref{cor:algo-1b-1c-DP},
shows that by modifying Algorithm 
\ref{alg:basic}, we can strictly enforce Conditions 
\ref{eq:downsampling probability2} and \ref{eq:artificial items bound} to guarantee $\epsilon$-DP.
Namely, by appropriately downsampling using $\pi_0 = 1-e^{-\epsilon}$, Algorithm 
\ref{alg:dp-large-set} is $\epsilon$-DP if we have an a prior guarantee that the number of items
in the stream is at least $n_0$.
If this guarantee is not known in advance, then the same $\pi_0$ is used for sampling, alongside 
the insertion of ``phantom'' elements to satisfy the minimum cardinality condition.

\begin{myproof}{Corollary}{\ref{cor:algo-1b-1c-DP}}{
Under their respective assumptions, Algorithms \ref{alg:dp-large-set} and \ref{alg:dp-any-set}, \texttt{DPSketchLargeSet} and \texttt{DPSketchAnySet} respectively,
satisfy Conditions \eqref{eq:downsampling probability2}  and \eqref{eq:artificial items bound}
 of Theorem \ref{thm:main result}.
}
\end{myproof}

\subsection{Private Sketches via Merging: Section \ref{sec:existing sketches post processing}}

Algorithm \ref{alg:make dp} converts a non-private sketch, $S_r$ into an $\epsilon$-DP sketch by merging 
it with a noisy sketch $T$.
Merging requires the same seed to be used so we suppress this notation in the subsequent writing.
The merge step is a property of the specific sketching algorithm used and operates on the sketch states $s$ and $s'$ so we
also overload the notation to denote the merge over states by $s \cup s'$.

Since sketch $s$ is already constructed, items cannot be first downsampled in the building phase the way they are in Algorithms \ref{alg:dp-large-set} and \ref{alg:dp-any-set}.
To achieve the stated privacy, Algorithm \ref{alg:make dp} constructs a noisily initialized sketch, $t$,
which satisfies both the downsampling condition (Condition \eqref{eq:downsampling probability2}) and the minimum 
stream cardinality requirement (Condition \eqref{eq:artificial items bound}) and returns the merged sketch 
$s \cup t$. 
Hence, the sketch will satisfy both conditions for $\epsilon$-DP, as shown in Corollary \ref{cor:makedp}

\begin{corollary}
Regardless of the sketch $s$ provided as input to the function
\texttt{MakeDP} (Algorithm \ref{alg:make dp}), $\texttt{MakeDP}$ yields an $\epsilon$-DP distinct counting sketch. 
\label{cor:makedp}
\end{corollary}

\begin{proof}
Given sketches $S,T$ with states $s$ and $t$, respectively,
we claim that any item that does not modify $T$ also cannot modify
the merged sketch $S \cup T$ by the order-invariance of $S, T$.
To see this, let $\dcal_S$ and $\dcal_T$ 
respectively denote the streams
that were processed by sketches $S$ and $T$,
and consider an item $i$ that does
not appear in $\dcal_S$ or $\dcal_T$
and whose insertion into $\dcal_T$ would
not change the sketch $T$. 
Since the state of the sketch $T$ is the same after processing $\dcal_T \circ i$
as it was after processing $\dcal_T$,
$S \cup T$ is also the sketch of $\dcal_T \circ i \circ \dcal_S$,
where $\circ$ denotes stream concatenation.
By order-invariance, $S \cup T$ is also a sketch
for $\dcal_T \circ \dcal_S\circ i$.
Also by order-invariance, 
$S \cup T$ is a sketch for $\dcal_T \circ \dcal_S$.
Hence, we have shown that the insertion of $i$
into $\dcal_T \circ \dcal_S$ does not change
the resulting sketch.

It follows that $\pi(s \cup t) \leq \pi(t) \leq \pi_0$, where the last inequality holds by the stopping
condition of the loop in \texttt{DPInitSketchForMerge} (Algorithm \ref{alg:DPInitSketchForMerge}). 
Hence, \texttt{MakeDP} also satisfies Condition \eqref{eq:downsampling probability2}.
The requirement that $v \geq n_0$ in \texttt{DPInitSketchForMerge} 
also ensures that $S \cup T$ is a sketch of a stream satisfying Condition \eqref{eq:artificial items bound}. Hence, Theorem \ref{thm:main result} implies
that the sketch $S \cup T$ returned by \texttt{MakeDP} satisfies $\epsilon$-DP.
Since the additional value $v$ that affects the estimate returned by \texttt{MakeDP} does not depend on the data, there is no additional privacy loss incurred by returning it.
\end{proof}

\section{Utility Proofs: Section \ref{sec:utility analysis}}
in this section we present the proofs showing that differentially private sketches have the same asymptotic performance as non-private sketches.
Namely, they remain unbiased and have the same variance as the number of unique items in the stream grows.
These results apply to Algorithms \ref{alg:basic}-\ref{alg:dp-any-set} and Algorithm \ref{alg:make dp}.
First we will show unbiasedness.

\begin{myproof}{Theorem}{\ref{thm:unbiased}}{
Trivially, Algorithm \ref{alg:basic} is unbiased by assumption, as it does not modify the base sketch.
Given $V$, there are $Z \sim Binomial(n + V, p)$ items added to the base sketch. 
Since the base sketch's estimator is unbiased, $\E(\hat{N}_{base} | Z) = Z$.
Algorithms \ref{alg:dp-large-set}, \ref{alg:dp-any-set}, and Algorithm \ref{alg:make dp} all have expectation:
\begin{align*}
    \E\, (\hat{N}(S_r(\dcal)) | V) &= \E \left( \E\left(\frac{\hat{N}_{base}}{p} - V \bigg| V,Z \right) \right) \\
    &= \E \left( \frac{Z}{p}-V \bigg| V \right) = n+V-V = n.
\end{align*}
}
\end{myproof}
Next we establish the variance properties of the sketching algorithms.
This involves expressing the variance of estimates from the algorithms in terms of the ``base sketch'' estimator.

\begin{corollary}
The variance of the estimates produced by Algorithms \ref{alg:basic}-\ref{alg:dp-any-set} and \ref{alg:make dp} is given by
\begin{align}
\Var \left( \hat{N}(S_r(\dcal)) \right)
&= \E \left( \Var \left(\frac{\hat{N}_{base}}{p} \bigg| V \right) \right).
\end{align}
\label{cor:var}
\end{corollary}
\begin{proof}
This follows from the law of total variance and the fact that the estimators are unbiased.
\end{proof}
Finally, we can leverage Corollary \ref{cor:var} to show there is no asymptotic increase in variance caused by the extra steps 
added to make a sketch private.

\begin{myproof}{Theorem}{\ref{thm:utility given relative error}}{
Let $Z \sim Binomial(n+v,\pi_0)$ denote the actual number of items inserted into the base sketch.
From Corollary \ref{cor:var} and since $V$ is constant, the variance is
\begin{align*}
    \Var\, \hat{N}(S_r(\dcal)) 
&= \left( \Var \left(\frac{\hat{N}_{base}}{\pi_0} \bigg| V=v \right) \right) \\
&= \left(\frac{\E\, \Var(\hat{N}_{base} | Z) + \Var\, \E(\hat{N}_{base} | Z)}{\pi_0^2}  \right) \\
&\leq \left(\frac{\E\, {Z^2}/{c} + \Var(Z)}{\pi_0^2}  \right) \\
&= \frac{(\E\,Z)^2}{c\pi_0^2} + \frac{\Var(Z)(c+1)}{c\,\pi_0^2}  \\
&= \frac{(n+v)^2}{c} + \frac{(n+v)(1-\pi_0)}{c \pi_0}.
\end{align*}
Trivially,  $\frac{\Var_n(\hat{N})}{\Var_n(\hat{N}_{base,n})} = \frac{(n+v)^2}{n^2} + O(1/n) \to 1 \,\,\mbox{ as } n \to \infty$.
}
\end{myproof}

\begin{corollary}
\label{cor: variance-ratios}
Assume that the conditions of Theorem \ref{thm:utility given relative error} hold. Further assume the base sketching algorithm satisfies an $(\epsilon, \delta_n)$ privacy guarantee where $\delta_n \to 0$ as $n \to \infty$. For any given $n^* > n_0$, we say Algorithm \ref{alg:make dp} succeeded if $V < n^*$. Then with probability at least $1-\delta_{n^*}$
\begin{align*}
\Var_n(\hat{N} | \mathit{Success}) &\leq \frac{(n+{n}^*)^2}{c} + \frac{(n+n^*)(n_0 + \pi_0^{-1})}{k_{max}}
\end{align*}
and
\begin{align*}
\frac{\Var_n(\hat{N}| V)}{\Var_n(\hat{N}_{base,n})} &\stackrel{p}{\to} 1 \,\,\mbox{ as } n \to \infty,
\end{align*}
where $X_n \stackrel{p}{\to} 1$ denotes convergence in probability: $P(|X_n - 1| < \Delta) \to 1$ as $n \to \infty$ for any $\Delta > 0$.
\end{corollary}

\section{Concrete Examples: Section \ref{sec:example-sketches} }
\label{app: sketch examples}

In this section we provide the proofs for results showing that popular sketches are $(\epsilon, \delta)$-DP.
We also provide further discussion for Adaptive Sampling that is omitted from the main body.

\subsection{FM85}

\begin{myproof}{Theorem}{\ref{thm:fm85-dp}}{
    To obtain an $(\epsilon, \delta)$ guarantee, note that 
    bit $s_{ij}$ in the sketch has probability $2^{-i}/k$ of 
    being selected by any item. 
    If $v = \lceil -\log_2 \pi_0 \rceil$ and all bits $s_{ij}$ with $j \leq v$ are flipped, then $\pi(s) < \pi_0$. 
    The probability $\Pr_r(s_{ij} = 0) = (1-2^{-i}/k)^n \leq \exp(-2^{-i}n/k)$.
    A union bound gives that
    $\Pr_r(\pi(\scal_r(\dcal)) \geq \pi_0) \leq k \sum_{i=1}^v \exp(-2^{-i}n/k) \leq kv \exp(-2^{-v}n/k) = kv \exp\left(- \tilde{\pi}_0 \frac{n}{k} \right)$
    where $\tilde{\pi}_0 = 2^{-v} \leq \pi_0$.
    }
\end{myproof}

Recall that the quantity $k_{max}$ for \acrshort{fm} is larger than all other sketches by either $32$ or $64$, the number of bits used in the hash function.
Thus, \acrshort{fm} requires a larger minimum number of items in the sketch to guarantee privacy, see Equation \eqref{eq:artificial items bound}.
However, the sketch is highly compressible as, for large $n$, each bitmap has 
entropy of approximately $4.7 $ bits \cite{lang2017back}. 
More recent works have placed this numerical result
on firmer theoretical footing
\cite{pettie2021information}, and in fact shown
that the resulting
space-vs.-error tradeoff is essentially \emph{optimal} amongst a large
class of sketching algorithms. 
A practical implementation of the compressed sketch can be found in the Apache DataSketches library \cite{datasketches}.\footnote{It is referred to as CPC, short for compressed probabilistic counting.}
It achieves close to constant update time by buffering
stream elements and only decompressing the sketch when the buffer is full.

Our results imply the above compressed sketches can yield a relaxed $(\epsilon, \delta_n)$-differential privacy guarantee when the number of inserted items is 
$n \geq n_0$ (Equation \eqref{eq:number-of-items-added}).
If the size of the sketch in bits is $b$, the sketch is $\epsilon$-differentially private if 
$n > \frac{b}{1-\exp(-\epsilon)}$ or equivalently
$b < n(1-\exp(-\epsilon))$.
Thus, $\delta_n = \Pr_r(b \geq n(1-\exp(-\epsilon))$.

\subsection{Linear Probabilistic Counting}
\label{app: lpca}

\begin{myproof}{Theorem}{\ref{thm:lpc-dp}}{
$S_r(D)$ is not privacy violating (i.e., $\pi(s) < \pi_0$) if
$\pi(\scal_r(\dcal)) = p(1-B/k) < \pi_0$ or, equivalently,
$B > k(1-\pi_0/p)$. Note that
$G_i \sim Geometric( p(1-i/k))$ items must be added for the number of 
filled bits to go from $i$ to $i+1$.

We can use a tail bound for the sum of geometric random variables \cite{janson2018tail}. Assume that $n \geq \frac{k-1}{1-\exp(-\epsilon)} \geq n_0$ so that Condition 
\ref{eq:artificial items bound} is satisfied. If $n > \mu_0$ then 
\begin{align}
    \delta &\leq P\left(\sum_{i=0}^{b_0} G_i > n\right) \\
    &\leq \exp\left(-\tilde{\pi}_0 (n/\mu_0 -1 -\log (n/\mu_0)) \right) \\
    &= \frac{\mu_0}{n} \exp\left(-\tilde{\pi}_0 (n/\mu_0 -1) \right). \\
\end{align}

The number of expected items needed to fill $b_0$ bits if $b_0 \geq 1$ is 
\begin{align}
    \mu_0 &:= \sum_{i=0}^{b_0-1} \frac{1}{p}\frac{1}{1-i/k} \\
    &< \frac{k}{p}  \left(\log\left(\frac{k}{k-b_0}\right) +1/(2 k) - 1/(2 (k-b_0)) +\frac{1}{12(k-b_0)^2}\right) \\
    &= \frac{k}{p} \log\left(\frac{k}{k-b_0}\right) - \frac{1}{p} \frac{b_0}{2(k-b_0)}
    + \frac{1}{p}\frac{k}{12(k-b_0)^2} \\
    &=  \frac{k}{p} \log\left(\frac{k}{k-b_0}\right) - \frac{1}{p} \frac{6 b_0(k-b_0) - k}{12(k-b_0)^2} \\
    &< \frac{k}{p} \log\left(\frac{k}{k-b_0}\right). 
\end{align}
}
\end{myproof}

\subsection{Bottom-$k$ / KMV}
\begin{myproof}{Theorem}{\ref{thm:kmv-dp}}{
The value $\pi(s)$ is equal to the $k^{th}$ minimum value.
If $X > k$, then the $k^{th}$ minimum value is $< \pi_0$ and Condition \ref{eq:downsampling probability2} is satisfied.
Thus, $\delta$ is the failure probability. The bound follows directly from Bernstein's inequality:
\begin{align}
    P(X \leq k) &= P(n - X > n - k) = P( (n-X) - n(1-\pi_0) > n \pi_0 - k) \\
    &\leq \frac{1}{2}\frac{(n\pi_0 - k)^2}{n\pi_0(1-\pi_0) + 1/3}\\
    &= \frac{1}{2} \frac{(\pi_0 - k/n)^2}{\pi_0(1-\pi_0) + \sfrac{1}{3n^2}} n
\end{align}
}
\end{myproof}

\subsection{Adaptive Sampling}
Wegman's adaptive sampling is similar to the bottom-$k$ sketch but does not require the sketch to store exactly $k$ hashes. 
Instead, it maintains a threshold $p$ and stores all hash values less than $p$.
Whenever the sketch size exceeds $k$, then the threshold is cut in half and only values under the threshold are retained. This ensures that processing a stream of length $n$ takes expected $O(n)$ time rather than $O(n \log k)$ as in a max-heap-based implementation of Bottom-$k$.

It is order invariant since the sketch only depends on the number of hash values under each of the potential thresholds and not the insertion order. 
Since at most $k$ hashes are stored, and the sketch is modified only if one of these hashes is removed, like KMV a.k.a. Bottom-$k$, 
the maximum number of items that can modify the sketch by removal is $k_{max} = k$.
In Corollary \ref{cor:ads-dp} we showed that Adaptive Sampling with $k$ buckets has the same privacy
behavior as a Bottom-$k$ sketch.

\begin{myproof}{Corollary}{\ref{cor:ads-dp}}{
Consider sketches $\scal^{AT}_r(\dcal)$, $\scal^{KMV}_r(\dcal)$ using the same hash function. 
Since the threshold in adaptive sampling is at most the $k^{th}$ minimum value, 
$\pi(\scal^{AT}_r(\dcal)) \leq \pi(\scal^{KMV}_r(\dcal))$.
So $\pi(\scal^{KMV}_r(\dcal)) < \pi_0 \implies \pi(\scal^{AT}_r(\dcal)) < \pi_0$.
}
\end{myproof}

\subsection{HyperLogLog}
\begin{myproof}{Theorem}{\ref{thm:hll-dp}}{
    In HLL, the sampling probability $\pi(s) = k^{-1} \sum_{i=1}^k 2^{-s_i}$ here $s_i$ is the value in each bin. 
    Thus, if all bins have value $s_i > -\log_2 \pi_0$, then $\pi(s) < \pi_0$.
    Let $C_i$ be the event that $s_i > -\log_2 \pi_0$.
    Then $P(\neg C_i | n) \leq (1-\pi_0 / k)^n \leq \exp(-\pi_0 n/k)$.
    A union bound gives $\Pr_r(\pi(S_r(\dcal)) \geq \pi_0) \leq k \exp(-\pi_0 n/k)$.
}
\end{myproof}

\begin{table}[]
\caption{
Comparison of the utility bounds for related work.
The only corresponding result to ours in \cite{choi2020differentially} is Theorem $4.2$, which shows that the 
LogLog sketch is $(\epsilon, \delta')$-DP provided $n \ge n_0'$ for $n'_0$ at least a factor $8$ larger than our 
$n_0$ from \eqref{eq:number-of-items-added}.
Our approaches, Algorithms \ref{alg:basic}-\ref{alg:dp-any-set} simultaneously achieve the best utility and update time,as well as a tighter privacy bound.
} 
\label{tab: related-work-bounds}
\centering
\begin{tabular}{@{}llll@{}}
\toprule
Algorithm                         & Privacy                                        & Utility (Relative Error): $\gamma$      & Update Time \\ \midrule
1a                                & $(\epsilon, \delta)$ for $n \ge n_0$           & $\frac{1.04}{\sqrt{k}}$                 & $O(1)$            \\
1b, 1c                            & $\epsilon$-DP for $n \ge n_0, 1$, respectively & $\frac{1.04}{\sqrt{k}}$                 & $O(1)$      \\
\cite{smith2020flajolet}          & $(\epsilon, 0)$ and $(\epsilon, \delta)$-DP    & $\frac{10 \ln(1/\beta)^{1/4}}{\sqrt{k}}$& $O(k)$      \\
\cite{choi2020differentially}     & $(\epsilon, \delta)$ for $n \ge n_0' \ge 8 n_0$& $\frac{1.3}{\sqrt{k}}$                 & $O(1)$      \\ \bottomrule
\end{tabular}
\label{tab:sketch-privacy-utility}
\end{table}

\section{Further Empirical Details}
\label{app:exptdetails}
\paragraph{Experiment 1: Update Time.}
We implement the regular, non-private HLL using a $32$-bit non-cryptographic MurmurHash. 
Our Private HLL (PHLL) is implemented in the Algorithm \ref{alg:dp-large-set} setting with a $256$-bit cryptographic hash function, SHA-$256$. 
In this model, PHLL employs the same algorithm as HLL but uses an extra downsampling step and rescales the estimator.
QLL, the $\epsilon$-DP algorithm of \cite{smith2020flajolet} also uses SHA-$256$.
We record the time to populate every sketch with $2^{10}$ updates with $k \in \{2^7, 2^8, \dots 2^{12} \}$ buckets.
Each marker represents the mean update time over all updates and the curves are the evaluated mean update time over $10$ trials.

We only implement \cite{smith2020flajolet} with $\gamma = 1.0$ because the running time is independent of $\gamma$.
Since the running time is independent of the privacy parameter for all methods, we test the total privacy budget $\epsilon = 1.0$.
All methods are implemented in Python so the speed performance could be optimized using a lower-level implementation.
In absolute terms, we find that for $2^{10}$ updates, HLL needs $8 \times 10^{-6}$ seconds compared to $9 \times 10^{-6}$ seconds for PHLL.
Both methods have a standard deviation of about $5 \times 10^{-7}$.
On the other hand, \cite{smith2020flajolet} needs between $2 \times 10^{-4}$ and $5 \times 10^{-3}$ seconds with a standard deviation on the order of $10^{-7}$ which is imperceptible on the scale of Figure \ref{fig:timing-comparison}.

\paragraph{Experiment 2: Utility-space tradeoff (Figure \ref{fig:utility-space}).}
We simulate each of the algorithms (HLL, PHLL, QLL) on an input stream of cardinality $n=2^{20}$ over $100$ independent trials.
The privacy budget is fixed at $\epsilon = \ln(2)$ meaning $\pi_0 = 1/2$ \eqref{eq:downsampling probability}.
The number of buckets was varied in  $k \in \{2^7, 2^8, \dots 2^{12} \}$ and we record the simulated 
$\textbf{total sketch size} \defn k \cdot \log_2 \max_i s_i$, which is the space for the largest bin value $s_i$ for a sketch with $k$ buckets. 
In this example, $k_{max} = k$ for HLL (Table \ref{tab:sketch comparison}) so that 
\eqref{eq:artificial items bound} $n_0 = 2(k - 1)$. 
As $n \ge n_0$ PHLL is $\epsilon$-DP in this setting by Theorem \ref{thm:main result}.
Essentially, PHLL under Algorithm \ref{alg:dp-large-set} employs HLL as the base sketch with a stream downsampled by $\pi_0 = 1/2$.

For QLL, the relative error $\gamma_{QLL} = \frac{10 \ln(1/\beta)^{0.25}}{\sqrt{k}}$ \cite[Theorem 2.6]{smith2020flajolet} depends on the 
sketching failure probability $\beta$ and the number of buckets.
This is a factor $10 \ln(1/\beta)^{0.25} / 1.04$ worse than our bound for $\gamma_{PHLL}$ and occurs with probability at least $1-\beta$.
We set $\beta = 0.05$ so that $\gamma_{QLL} \approx 7.49 / \sqrt{k}$ for \cite{smith2020flajolet},
while our method has $\gamma_{PHLL} = 1.04 / \sqrt{k}$.
Only the base-$(1+\gamma_{QLL})$ harmonic estimator was tested for closest comparison to HLL and our work.

Figure \ref{fig:utility-space} plots the utility in relative error against the total space usage of the methods.
We see that PHLL is indistinguishable from HLL.
The utility of QLL appears comparable to (P)HLL and is better than its worst-case relative error guarantees, 
yet this comes at the cost of using more space than our sketch, PHLL.
In absolute terms PHLL consumes approximately $90\%$ of the space used by QLL when $k=2^7$, which decreases to roughly $65\%$ when $k=2^{12}$.
The fractional reduction in space usage of PHLL over QLL is because $\gamma$ decreases as $k$ grows
and reducing $\gamma$ affects the $Geometric(\frac{\gamma}{1+\gamma})$ hash function used to
select the bin values in QLL; smaller $\gamma$ result in larger bin values which inflate the size of the QLL sketch.
In fact, we find that the mean total space usage of QLL compared to PHLL, is larger by a factor of $O(\log k)$.
This agrees with the ratio of theoretical space bounds of \cite{smith2020flajolet} and our work, as illustrated in Figure \ref{fig:space-ratio}.

\begin{figure*}[b]
    \centering
    \input{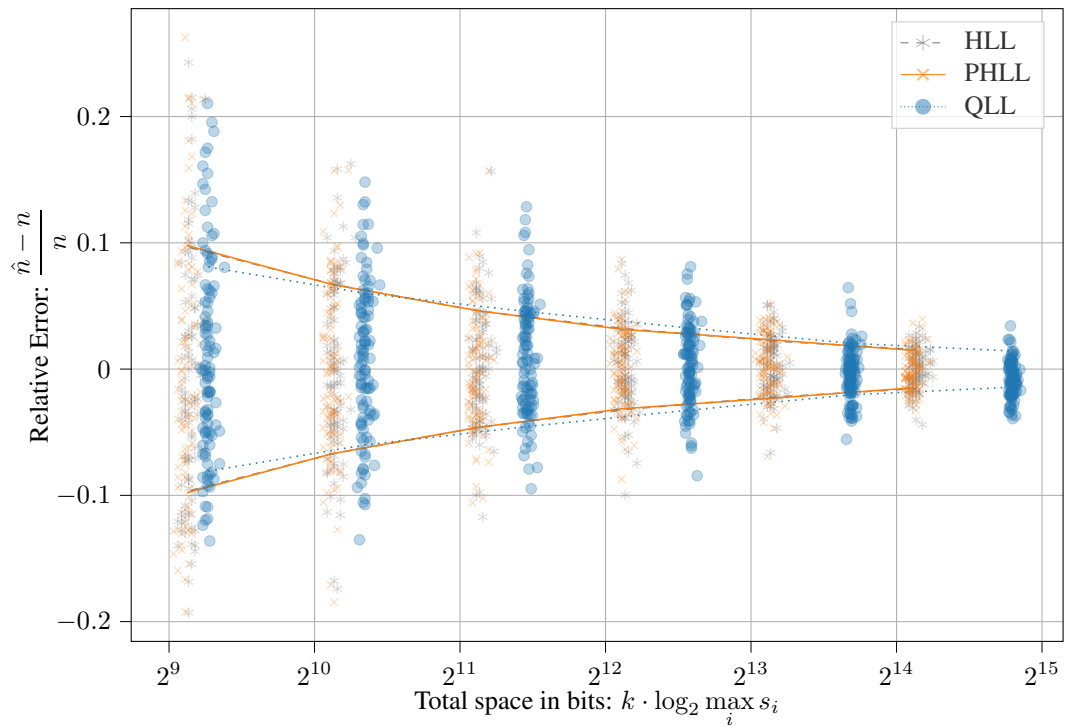}
    \caption{  
    The curves represent empirical standard deviations of the estimates.
    For our method, PHLL, this matches the error bound $1.04 / \sqrt{k}$ as indicated from Table 
    \ref{tab:sketch comparison}.
    Empirically, QLL has a nearly matching standard deviation to PHLL, despite a suboptimal utility
    bound as seen in Table \ref{tab:sketch-privacy-utility}.}
    \label{fig:utility-space}
\end{figure*}

\end{document}